\documentclass[12pt]{article}
\usepackage{fullpage}
\usepackage{amsfonts}
\usepackage{amsmath}
\usepackage{amssymb}
\usepackage{mathtools}
\usepackage{graphicx}
\usepackage{caption}
\usepackage{subcaption}
\usepackage{float}
\usepackage{dsfont}
\usepackage{physics}
\usepackage{bbold}
\usepackage{csquotes}
\usepackage{tabularx}
\usepackage{multirow}
\usepackage{pbox}
\usepackage[utf8]{inputenc}
\usepackage[T1]{fontenc}
\usepackage[absolute]{textpos}
\usepackage[dvipsnames]{xcolor}
\usepackage[colorlinks=true,linkcolor=blue,citecolor=magenta,plainpages=false,pdfpagelabels]
{hyperref}
\usepackage{verbatim}
\usepackage{newtxtext,newtxmath}

\allowdisplaybreaks
\providecommand{\U}[1]{\protect\rule{.1in}{.1in}}
\newcommand{\cmmnt}[1]{}
\newtheorem{theorem}{Theorem}

\newtheorem{corollary}{Corollary}

\newtheorem{definition}{Definition}

\newtheorem{lemma}{Lemma}

\newtheorem{proposition}{Proposition}
\newtheorem{remark}{Remark}

\newenvironment{proof}[1][Proof]{\noindent\textbf{#1.} }{\ \rule{0.5em}{0.5em}}
\allowdisplaybreaks

\DeclareMathOperator{\Ext}{Ext}

\makeatletter
\let\@fnsymbol\@arabic
\makeatother

\begin{document}

\title{\vspace{-.3in}\textbf{Limiting  one-way distillable secret key \\via privacy testing of extendible states}}

  \author{Vishal Singh\thanks{Mathematical Quantum Information RIKEN Hakubi Research Team, 
RIKEN Pioneering Research Institute (PRI) and RIKEN Center for Quantum Computing (RQC), Wako, Saitama 351-0198, Japan \\ (email: vishalsph04@gmail.com)}\hspace{0.3em} \thanks{Centre for Quantum Technologies, National University of Singapore, Singapore 117543, Singapore}\hspace{0.3em} \thanks{School of Applied and Engineering Physics, Cornell University, Ithaca, New York 14850, USA}
  \and
  Karol Horodecki\thanks{Institute of Informatics, National Quantum Information Centre,
Faculty of Mathematics, Physics and Informatics,
University of Gda\'nsk, Wita Stwosza 57, 80-308 Gda\'nsk, Poland}
\and
  Aby Philip\thanks{Institute of Fundamental Technological Research, Polish Academy of Sciences, Pawi\'nskiego 5B, 02-106 Warsaw, Poland.}
  \and
  Mark M.~Wilde\thanks{School of Electrical and Computer Engineering, Cornell University, Ithaca, New York 14850, USA}
  }

\date{ }
\maketitle

\begin{abstract}
    The notions of privacy tests and $k$-extendible states have both been instrumental in quantum information theory, particularly in understanding the limits of secure communication. In this paper, we determine the maximum probability with which an arbitrary $k$-extendible state can pass a privacy test, and we prove that it is equal to the maximum fidelity between an arbitrary $k$-extendible state and the standard maximally entangled state. Our findings, coupled with the resource theory of $k$-unextendibility, lead to an efficiently computable upper bound on the one-shot, one-way distillable key of a bipartite state, and we prove that it is equal to the best-known efficiently computable upper bound on the one-shot, one-way distillable entanglement. We also establish efficiently computable upper bounds on the one-shot, forward-assisted private capacity of channels. Extending our formalism to the independent and identically distributed setting, we obtain single-letter efficiently computable bounds on the $n$-shot, one-way distillable key of a state and the $n$-shot, forward-assisted private capacity of a channel. For some key examples of interest, our bounds are significantly tighter than other known efficiently computable bounds. 
\end{abstract}


\tableofcontents

\section{Introduction}

Monogamy of entanglement is a unique feature of quantum correlations with no analog in classical probability theory~\cite{CKW00,Ter04}. Roughly stated, if two systems are highly entangled, then neither of them can be equally entangled with a third system. Not only has this property been a subject of fundamental interest in quantum information theory~\cite{KW04,OV06,dOCF14,BT24}, but it also is the vital feature that allows unconditional secure communication over a quantum network~\cite{E91, Paw10}.

The discovery of the first unconditionally secure communication protocol in~\cite{BB84} has led to a consolidated effort from the quantum information theory community to understand the connection between shared entanglement and the ability to perform unconditionally secure communication between two distant parties. Significant foundational developments were made in~\cite{HHHO05, HHHO09}, where the authors identified the mathematical structure of all bipartite states that yield a secret key upon local measurements, now known as ``\textit{private states}''. In~\cite{HHHLO08, HHHLO08_QP}, the authors devised the ``\textit{privacy test}'', a bipartite measurement that determines if a given state is \textit{private} or not. Determining the probability of a state to pass the privacy test has led to several insights into the theory of secure communication over a quantum network \cite{WTB17}.

On the other hand, symmetric extendibility of states has emerged as a powerful tool for understanding the limits of quantum information processing, capturing the notion of monogamy of entanglement~\cite{Wer89, DPS02,TDS03, DPS04}. It has been especially useful in understanding the limits of entanglement distillation and key distillation from a bipartite state under local operations and one-way classical communication, which we abbreviate as one-way LOCC in the remainder of this paper. However, a fundamental question has remained unanswered until now:
\begin{displayquote}
\textit{What is the maximum probability with which a symmetrically extendible state can pass the privacy test?}
\end{displayquote}
We definitively answer this question in our paper and use it to obtain efficiently computable bounds on several quantities of interest in the theory of secure communication over quantum channels.

We particularly focus on $k$-extendibility of states defined in~\cite{DPS02,DPS04}, which is a hierarchy of semidefinite conditions for testing the separability of a bipartite state. A bipartite state is separable if and only if it is $k$-extendible for every integer $k\ge 2$. Beyond its use as a relaxation of the separability criterion, the $k$-extendibility conditions identify a set of states, strictly larger than the set of separable states, that are useless for the task of entanglement distillation or key distillation using one-way LOCC protocols. This has motivated their study in a resource-theoretic framework~\cite{KDWW19,KDWW21}.

The $k$-extendibility of states was first studied from a resource-theoretic approach in~\cite{KDWW19,KDWW21}, where the authors defined the resource theory of $k$-unextendibility and used it to study entanglement transmission over quantum channels assisted by local operations and forward classical communication. Central to their developments was the fact that the fidelity between a $d$-dimensional, $k$-extendible state and the maximally entangled state of Schmidt rank $d$ cannot exceed $\frac{1}{d} + \frac{1}{k} - \frac{1}{dk}$. We show that this quantity is also the maximum probability with which any $k$-extendible state can pass the privacy test for $\log_2 d$ secret bits. As a consequence, we show that all the limits obtained in~\cite{KDWW19,KDWW21} on forward-assisted quantum communication tasks are, in fact, limits on the corresponding forward-assisted private communication tasks as well. 

We use the framework of the resource theory of $k$-unextendibility to study the ultimate limits of private communication over a quantum network. In what follows, we briefly discuss applications of our results to the task of secret-key distillation from bipartite states under one-way LOCC and to the task of private communication over a channel with local operations and forward public communication.

\subsection{Secret key distillation from states}

The task of secret-key distillation from states using local operations and an arbitrary amount of one-way public communication was studied in~\cite{DW05} in the asymptotic setting. Further studies extended the results to the non-asymptotic setting~\cite{RR12,KKGW21}, where two parties wish to establish a highly secure key, and not necessarily perfectly secure key, using a finite number of independent and identically distributed (i.i.d.)~states along with local operations and one-way public communication. This quantity is termed as the one-shot, one-way distillable key of a state. 

While previous works have obtained lower and upper bounds on the one-shot, one-way distillable key of a state, an efficiently computable upper bound on the one-shot, one-way distillable key of a state was only recently discovered~\cite{SW25}. However, the recent bound from~\cite{SW25} does not give a meaningful value if the error tolerance in the key distillation protocol is allowed to be too high or if the resource state is highly entangled.

Several upper bounds on the one-shot, one-way distillable key of a state, which is the number of secret bits that can be distilled from a state using local operations with public communication from both parties, have also been discovered, such as the hypothesis-testing relative entropy of entanglement bound~\cite{WTB17} and the squashed entanglement bound~\cite{Wilde16}. Naturally, these bounds serve as upper bounds on the one-shot, one-way distillable key of a state as well. However, neither the hypothesis-testing relative entropy of entanglement nor the squashed entanglement of a state are efficiently computable in general.

We obtain a new efficiently computable upper bound on the one-shot, one-way distillable key of a state, and we numerically demonstrate that our bound is tighter than the bound in~\cite{SW25} for isotropic states, as an example. Moreover, our methods allow us to obtain a family of upper bounds on the one-shot, one-way distillable key of a state. While some of the bounds in this family may not give a meaningful result for a given error tolerance, there always exists a bound in this family that yields a meaningful limit on the one-shot, one-way distillable key of the state. The hypothesis-testing relative entropy of entanglement bound~\cite{WTB17} appears as a limiting case of this family of bounds.

As stated earlier, our bounds on the one-shot, one-way distillable key are precisely equal to the bounds on the one-shot, one-way distillable entanglement of a state obtained in~\cite{KDWW19,KDWW21}, which are the best-known efficiently computable bounds on the one-shot, one-way distillable entanglement of a state to the best of our knowledge.

In a practical key distillation protocol, we often see that a large number of resourceful states are consumed before a single secret bit can be distilled with the desired security parameter. We use our methods to compute a lower bound on the minimum number of copies of an isotropic state that are needed to distill a single secret bit via a one-way LOCC protocol.

\subsection{Private communication over channels}

The notion of private capacity of a channel was first studied in~\cite{Dev05, CWY04}. Several developments in the study of private communication over a finite number of channel uses have been made in the last two decades~\cite{RR11, WTB17, Wilde17, RSW17, KKGW21}. An efficiently computable upper bound on the number of bits that can be securely transmitted over a single use of a channel assisted by local operations and forward public communication, which is termed as the one-shot, forward-assisted private capacity of the channel, was also discovered in~\cite{SW25}.

Here we introduce a new efficiently computable upper bound on the one-shot, forward-assisted private capacity of channels, and we numerically demonstrate that it is significantly tighter than the bound from~\cite{SW25} for erasure channels. Similar to the case of key distillation, we obtain a family of bounds on the one-shot, forward-assisted private capacity of a channel, and the hypothesis-testing relative entropy of entanglement of a channel appears as a limiting case of this family of bounds.

In applying our results to the one-shot, forward-assisted private capacity of channels, we define a new monotone for the resource theory of $k$-unextendibility of point-to-point channels, namely the $k$-unextendible generalized divergence of channels, which differs from the one considered in~\cite{KDWW19,KDWW21}. We take this slightly different approach in order to ensure that our bounds can be efficiently computed for every channel, which may not necessarily be the case for the monotone considered in~\cite{KDWW19,KDWW21}.

We also compute the minimum number of uses of an erasure channel needed to transmit a single bit securely over these channels when assisted by local operations and an arbitrary amount of forward public communication. 

\subsection{Summary of results}

The main technical result of this paper is a tight upper bound on the maximum probability with which a $k$-extendible state can pass a privacy test. We then use this result to obtain the following:
\begin{itemize}
    \item \textbf{(Theorem~\ref{theo:dd_1shot_key_bnd_k_ext})} Upper bound on the one-shot, one-way distillable key of a state, which can be computed using a semidefinite program.
    \item \textbf{(Corollary~\ref{cor:DD_key_n_shot_ub_sandwich})} Single-letter upper bound on the $n$-shot, one-way distillable key of a state, which can be computed using a semidefinite program.
    \item \textbf{(Theorem~\ref{theo:priv_cap_ub_hypo_test})} Upper bound on the one-shot, forward-assisted private capacity of a channel, which can be computed using a semidefinite program.
    \item \textbf{(Corollaries~\ref{cor:n_shot_priv_cap_ub} and~\ref{cor:n_shot_priv_cap_sandwich})} Single-letter upper bound on the $n$-shot, forward-assisted private capacity of a channel. The upper bound in Corollary~\ref{cor:n_shot_priv_cap_ub} can be computed using a semidefinite program.
\end{itemize}

The remainder of this paper is organized as follows:
\begin{itemize}
    \item In Section~\ref{sec:background}, we briefly review the notation used in this paper, the notion of secret keys and private states, and the resource theory of $k$-unextendibility.
    \item In Section~\ref{sec:priv_test}, we present the main technical result of this paper, which is a tight upper bound on the maximum probability with which a $k$-extendible state can pass a privacy test.
    \item In Section~\ref{sec:one_way_key_distill}, we review the notion of one-shot, one-way distillable key of a state, and we obtain efficiently computable upper bounds on the one-shot and $n$-shot, one-way distillable key of a state.
    \item In Section~\ref{sec:one_way_priv_cap}, we review the notion of one-shot, forward-assisted private capacity of a channel, and we obtain efficiently computable upper bounds on the one-shot and $n$-shot, forward-assisted private capacity of a channel.
    \item In Section~\ref{sec:applications}, we demonstrate numerical values of the upper bounds on the $n$-shot, one-way distillable key of isotropic states, and the $n$-shot, forward-assisted private capacity of erasure channels obtained from our bounds. We also compute a lower bound on the minimum number of isotropic states needed to distill a single secret bit using a one-way LOCC protocol for a fixed error tolerance. Similarly, we compute a lower bound on the minimum number of uses of an erasure channel to transmit a single bit over the channels with one-way LOCC assistance.  
\end{itemize} 

\section{Background}\label{sec:background}

In this section, we review some background material that is necessary to discuss the findings of this paper.

\subsection{Quantum states and channels}
A quantum state $\rho_A$ is a positive semidefinite, unit-trace operator acting on the Hilbert space $\mathcal{H}_A$ associated with the system $A$. We denote the set of all states acting on $\mathcal{H}_A$ by $\mathcal{S}(A)$, and we denote the dimension of $\mathcal{H}_A$ by $|A|$.

A bipartite state $\rho_{AB}$ acting on the Hilbert space $\mathcal{H}_A\otimes \mathcal{H}_B$ is said to be separable if it can written in the following form:
\begin{equation}
    \rho_{AB} = \sum_{x\in \mathcal{X}}p_x\sigma^x_A\otimes \tau^x_B,
\end{equation}
where $\mathcal{X}$ is an alphabet of arbitrary dimension, $\left\{p_x\right\}_{x\in \mathcal{X}}$ is a probability distribution, and $\left\{\sigma^x_A\right\}_{x\in \mathcal{X}}$ and $\left\{\tau^x_A\right\}_{x\in \mathcal{X}}$ are sets of quantum states. If a state is not separable, it is said to be \textit{entangled}. The maximally entangled state on the Hilbert space $\mathcal{H}_A\otimes \mathcal{H}_B$, with $|A| = |B|$, is defined as follows:
\begin{equation}
    \Phi^d_{AB} \coloneqq \frac{1}{d}\sum_{i,j=0}^{d-1}|i\rangle\!\langle  j|_A\otimes |i\rangle\!\langle j|_{B},
\end{equation}
where $\left\{|i\rangle\right\}_{i=0}^{d-1}$ is an orthonormal basis for both $\mathcal{H}_A$ and $\mathcal{H}_B$, and $d = |A| = |B|$ is the Schmidt rank of the maximally entangled state.

We often consider states acting on several isomorphic systems. To declutter the expressions, we use the following notation:
\begin{equation}
    B_{[k]} \coloneqq B_1B_2\cdots B_k,
\end{equation}
and we use the symbol $B_{[k]\setminus i}$ to describe the collection of systems $B_{[k]}$ but without system $B_i$. That is,
\begin{equation}
    B_{[k]\setminus i} \coloneqq B_1B_2\cdots B_{i-1}B_{i+1}B_{i+2}\cdots B_k.
\end{equation}

A quantum channel $\mathcal{N}_{A\to B}$ is a completely positive (CP), trace-preserving (TP) map that transforms a linear operator acting on $\mathcal{H}_A$ to a linear operator acting on $\mathcal{H}_B$. A channel is completely described by its Choi operator, which is defined as follows:
\begin{equation}\label{eq:Choi_op_defn}
    \Gamma^{\mathcal{N}}_{RB} \coloneqq \mathcal{N}_{A\to B}\!\left(d\Phi^d_{RA}\right),
\end{equation}
where system $R$ is isomorphic to system $A$ (denoted as $R\cong A$) and $d = |R| = |A|$. Rescaling the Choi operator to obtain a unit trace operator, we get the Choi state of the channel, which is defined as follows:
\begin{equation}
    \Phi^{\mathcal{N}}_{RB} \coloneqq \mathcal{N}_{A\to B}\!\left(\Phi^d_{RA}\right) = \frac{1}{|R|}\Gamma^{\mathcal{N}}_{RB}.
\end{equation}
We denote the set of all channels by $\operatorname{CPTP}$.

Channels that can be simulated by local operations and one-way classical communication are called one-way LOCC channels. An arbitrary bipartite one-way LOCC channel can be mathematically expressed as follows:
\begin{equation}
    \mathcal{L}^{\to}_{AB\to A'B'} = \sum_{x\in \mathcal{X}}\mathcal{E}^x_{A\to A'}\otimes\mathcal{F}^x_{B\to B'}, 
\end{equation}
where $\left\{\mathcal{E}^x_{A\to A'}\right\}_{x\in \mathcal{X}}$ is a set of completely positive maps such that $\sum_{x\in \mathcal{X}}\mathcal{E}^x_{A\to A'}$ is a quantum channel and $\left\{\mathcal{F}^x_{B\to B'}\right\}_{x\in \mathcal{X}}$ is a set of quantum channels.

\subsection{Secret keys and private states}

In this section, we review the notion of secret keys, private states, and the privacy test.

A $d$-dimensional tripartite key state is a classical-classical-quantum (ccq) state of the following form:
\begin{equation}
    \omega^d_{ABE} = \frac{1}{d}\sum_{i=0}^{d-1}|i\rangle\!\langle i|_A\otimes|i\rangle\!\langle i|_B\otimes \tau_E.
\end{equation}
When sharing a tripartite key state, Alice and Bob's classical symbols are perfectly correlated and uniformly random. Moreover, the eavesdropper's state is independent of Alice and Bob's systems, making it impossible for the eavesdropper to learn anything about them. The objective of any key distillation protocol is to ensure that the tripartite state shared between Alice and Bob and any possible eavesdropper is a tripartite key state.

A striking reduction from the tripartite picture of key distillation to a bipartite picture was discovered in~\cite{HHHO05,HHHO09}. In particular, any bipartite quantum state that yields $\log_2 d$ secret bits upon local measurements can be mathematically expressed in the following form:
\begin{equation}\label{eq:priv_st_defn}
    \gamma^d_{ABA'B'} = V_{ABA'B'}\!\left(\Phi^d_{AB}\otimes \tau_{A'B'}\right)V^{\dagger}_{ABA'B'},
\end{equation}
where $\Phi^d_{AB}$ is the maximally entangled state of Schmidt rank $d$, $\tau_{A'B'}$ is an arbitrary quantum state, and $V_{ABA'B'}$ is a unitary of the following form:
\begin{equation}\label{eq:twist_unitary_defn}
    V_{ABA'B'} = \sum_{i=0}^{d-1}I_A\otimes |i\rangle\!\langle i|_B\otimes U^i_{A'B'},
\end{equation}
with $\left\{U^i_{A'B'}\right\}_{i=0}^{d-1}$ being a set of arbitrary unitaries. Any state that is of the form given in~\eqref{eq:priv_st_defn} is called a private state of $\log_2 d$ secret bits. As such, the task of distilling secret keys is equivalent to the task of distilling private states from a shared bipartite state using a given set of operations, which is the set of one-way LOCC channels in this work. 

Distilling exact private states from a mixed state using one-way LOCC channels can be challenging. In fact, it is impossible to distill exact private states, even probabilistically, from commonly encountered states, such as Werner states and erased states~\cite{SW24_prob}, using one-way LOCC channels. In practice, we are often satisfied even if the distilled state is close to a private state with respect to some error tolerance $\varepsilon$. In this work, we follow~\cite{WTB17} and quantify the closeness of a bipartite state with a private state using fidelity of states, as defined below:
\begin{equation}
    F(\rho,\sigma) \coloneqq \left(\operatorname{Tr}\!\left[\sqrt{\sqrt{\sigma}\rho\sqrt{\sigma}}\right]\right)^2.
\end{equation}
This criterion of \textit{approximate} secret key distillation is also consistent with the usual notion of quantifying the error in key distillation by the trace distance between the final ccq state obtained after the protocol and an ideal tripartite key state, as argued in~\cite[Appendix C]{HHHO09}. 

One can test if a bipartite state is of the form given in~\eqref{eq:priv_st_defn} using the privacy test described by the POVM $\left\{\Pi^{\gamma}, I - \Pi^{\gamma}\right\}$~\cite{HHHLO08,HHHLO08_QP}, where
\begin{equation}\label{eq:priv_test_defn}
    \Pi^{\gamma}_{ABA'B'} \coloneqq V_{ABA'B'}\!\left(\Phi^d_{AB}\otimes I_{A'B'}\right)V^{\dagger}_{ABA'B'}
\end{equation}
and $V_{ABA'B'}$ is determined by the specific private state $\gamma^d_{ABA'B'}$ being tested for. Clearly, the state $\gamma^d_{ABA'B'}$ passes the privacy test with probability equal to one. Moreover, if for a given state $\omega_{ABA'B'}$,
\begin{equation}
    F\!\left(\omega_{ABA'B'}, \gamma^d_{ABA'B'}\right) \ge 1-\varepsilon,
    \label{eq:fid-to-priv-state}
\end{equation}
then the probability of $\omega_{ABA'B'}$ passing the privacy test is at least $1-\varepsilon$~\cite[Lemma 9]{WTB17}. That is, Eq.~\eqref{eq:fid-to-priv-state} implies that
\begin{equation}\label{eq:priv_test_pass_prob}
    \operatorname{Tr}\!\left[\Pi^{\gamma}_{ABA'B'}\omega_{ABA'B}\right] \ge 1-\varepsilon.
\end{equation}

\subsection{\texorpdfstring{$k$}{k}-Extendibility}

In this section, we briefly review the resource theory of $k$-unextendibility, which was developed in~\cite{KDWW19,KDWW21}, and is central to the results developed later in Sections~\ref{sec:one_way_key_distill} and~\ref{sec:one_way_priv_cap}.

For every integer $k\ge 2$, the resource theory of $k$-unextendibility comprises of $k$-extendible states as free states and $k$-extendible channels as free channels. 

The set of $k$-extendible states is defined as follows:
\begin{equation}
    \operatorname{Ext}_k\!\left(A\!:\!B\right) \coloneqq \left\{\begin{array}{cc}
        \sigma_{AB} \in \mathcal{S}(AB):  \\
          \exists ~\omega_{AB_{[k]}}\ge 0,\,  \operatorname{Tr}_{B_{[k]\setminus 1}}\!\left[\omega_{AB_{[k]}}\right] = \sigma_{AB},\\
          W^{\pi}_{B_{[k]}}\omega_{AB_{[k]}}\left(W^{\pi}_{B_{[k]}}\right)^{\dagger} = \omega_{AB_{[k]}} \qquad \forall \pi \in S_k
    \end{array} \right\},
\end{equation}
where $W^{\pi}_{B_{[k]}}$ is a unitary operator that permutes the systems $B_{[k]}$ according to the permutation $\pi$ in the symmetric group $S_k$.

It can be easily shown that every separable state is $k$-extendible for every $k\ge 2$. In fact, a bipartite state is separable if and only if it is $k$-extendible for every $k\ge 2$~\cite{DPS04} (see also \cite[Theorem~II.7]{CKMR08}). While testing the separability of a bipartite state is a hard problem~\cite{Gur03, Gha10}, one can test if a state is $k$-extendible for a fixed $k$ using a semidefinite program that scales polynomially with the dimension of the systems. Hence, the set of $k$-extendible states is a semidefinite relaxation of the set of separable states.

The free operations in the resource theory of $k$-unextendibility, as proposed in~\cite{KDWW19,KDWW21}, are $k$-extendible channels. A bipartite channel $\mathcal{N}_{AB\to A'B'}$ is said to be $k$-extendible if there exists a channel $\mathcal{P}_{AB_{[k]}\to A'B'_{[k]}}$ such that the following conditions are satisfied:
\begin{align}
    \operatorname{Tr}_{B'_{[k]\setminus 1}}\circ\mathcal{P}_{AB_{[k]}\to A'B'_{[k]}} &= \mathcal{N}_{AB\to A'B'}\otimes\operatorname{Tr}_{B_{[k]\setminus 1}},\label{eq:marginal_cond_k_ext}\\
    \mathcal{W}^{\pi}_{B'_{[k]}}\circ\mathcal{P}_{AB_{[k]}\to A'B'_{[k]}} &= \mathcal{P}_{AB_{[k]}\to A'B'_{[k]}} \circ\mathcal{W}^{\pi}_{B_{[k]}} \qquad \forall \pi \in S_k,\label{eq:perm_cov_cond_k_ext}
\end{align}
where $\mathcal{W}^{\pi}$ is the unitary channel, $\mathcal{W}^{\pi}(\cdot) = W^{\pi}(\cdot)\left(W^{\pi}\right)^{\dagger}$, corresponding to the permutation $\pi$ in the symmetric group $S_k$. The action of a $k$-extendible channel results in a $k$-extendible state, which justifies their treatment as free operations in the resource theory of $k$-unextendibility.

Every one-way LOCC channel is $k$-extendible for every $k\ge 2$. As such, the set of $k$-extendible channels can be viewed as a semidefinite relaxation of the set of one-way LOCC channels.

In~\cite{KDWW19, KDWW21}, the authors defined the $k$-unextendible divergence of a state, which serves as a resource monotone in the resource theory of $k$-unextendibility. Let $\mathbf{D}$ be a generalized divergence~\cite{PV10}. Then the $k$-unextendible generalized divergence of a state $\rho_{AB}$ is defined as follows:
\begin{equation}\label{eq:gen_k_unext_defn}
    \mathbf{E}_k\!\left(\rho_{AB}\right) \coloneqq \inf_{\sigma_{AB}\in \operatorname{Ext}_k(A:B)}\mathbf{D}\!\left(\rho_{AB}\middle\Vert\sigma_{AB}\right).
\end{equation}
The $k$-unextendible divergence of a state has the following properties, justifying its use as a resource monotone in the resource theory of $k$-unextendibility:
\begin{enumerate}
    \item The $k$-unextendible generalized divergence decreases monotonically under the action of bipartite $k$-extendible channels. That is,
\begin{equation}\label{eq:unext_ent_dat_proc}
    \mathbf{E}_k\!\left(\rho_{AB}\right) \geq \mathbf{E}_k\!\left(\mathcal{N}_{AB\to A'B'}\!\left(\rho_{AB}\right)\right) , 
\end{equation}
    for every $k$-extendible channel $\mathcal{N}_{AB\to A'B'}$. We refer the reader to~\cite{KDWW21} for a proof and further details.
    
    \item The $k$-unextendible generalized divergence of a $k$-extendible state is equal to the minimum value of the underlying divergence acting on an arbitrary pair of states. Consequently, the $k$-unextendible generalized divergence vanishes for $k$-extendible states. Furthermore, if the underlying divergence is faithful, then the induced $k$-unextendible divergence of a state is equal to zero if and only if the state is $k$-extendible. 
\end{enumerate}

In this work, we employ the $k$-unextendible divergence induced by the hypothesis-testing relative entropy, the $\alpha$-sandwiched R\'enyi relative entropy for $\alpha \in (1,\infty)$, and the $\alpha$-geometric R\'enyi relative entropy for $\alpha \in (1,2]$. We discuss these quantities here briefly.

\subsubsection{\texorpdfstring{$k$}{k}-Unextendible hypothesis testing divergence}

The hypothesis testing relative entropy between states $\rho$ and $\sigma$, also known as smooth-min relative entropy, is defined for a parameter $\varepsilon\in [0,1]$ as follows~\cite{BD10,BD11,WR12}:
\begin{equation}\label{eq:hypo_test_defn}
    D^{\varepsilon}_H\!\left(\rho\Vert\sigma\right) \coloneqq -\log_2 \inf_{0\le \Lambda \le I}\left\{\operatorname{Tr}\!\left[\Lambda\sigma\right]: \operatorname{Tr}\!\left[\Lambda\rho\right]\ge 1-\varepsilon\right\}.
\end{equation}
The $k$-unextendible hypothesis testing divergence is then defined as follows:

\begin{equation}\label{eq:unext_ent_hypo_test}
        E^{\varepsilon}_k\!\left(\rho_{AB}\right) = \inf_{\sigma_{AB}\in \operatorname{Ext}_k(A:B)}D^{\varepsilon}_H\!\left(\rho_{AB}\Vert\sigma_{AB}\right).
\end{equation}

The $k$-unextendible hypothesis testing divergence can be computed using a semidefinite program. See Appendix~\ref{app:SDPs} for the explicit semidefinite program.

\subsubsection{\texorpdfstring{$k$}{k}-Unextendible sandwiched R\'enyi divergence}

\label{sec:sandwich_k_unext_st}

The $\alpha$-sandwiched R\'enyi relative entropy between a state $\rho$ and a positive semidefinite operator $\sigma$ is defined for a parameter $\alpha \in \left[\frac{1}{2},1\right)\cup (1,\infty)$ as follows~\cite{MDSST13, WWY14}:
\begin{equation}
    \widetilde{D}_{\alpha}\!\left(\rho\Vert\sigma\right) \coloneqq \frac{1}{\alpha - 1}\log_2\operatorname{Tr}\!\left[\left(\sigma^{\frac{1-\alpha}{2\alpha}}\rho\sigma^{\frac{1-\alpha}{2\alpha}}\right)^{\alpha}\right].
\end{equation}
The $k$-unextendible sandwiched R\'enyi divergence of a state is then defined as follows:
\begin{equation}\label{eq:sandwich_unext_ent}
    \widetilde{E}^{\alpha}_k\!\left(\rho_{AB}\right)\coloneqq \inf_{\sigma_{AB}\in \operatorname{Ext}_k(A:B)}\widetilde{D}_{\alpha}\!\left(\rho_{AB}\Vert\sigma_{AB}\right) \qquad \forall \alpha \in \left[\frac{1}{2},1\right)\cup(1,\infty).
\end{equation}

The $k$-unextendible sandwiched R\'enyi divergence has several desirable properties. Here we note some key properties that we use in this paper and we refer the reader to~\cite{KDWW21} for further reading:
\begin{enumerate}
    \item \textbf{Subadditivity:} The $k$-unextendible sandwiched R\'enyi divergence is subadditive under tensor products for every $k\ge 2$ and every $\alpha \in \left[\frac{1}{2},1\right)\cup(1,\infty)$. That is,
    \begin{equation}
        \widetilde{E}^{\alpha}_k\!\left(\rho_{AB}\otimes \sigma_{CD}\right) \le \widetilde{E}^{\alpha}_k\!\left(\rho_{AB}\right) + \widetilde{E}^{\alpha}_k\!\left(\sigma_{CD}\right),
    \end{equation}
    where $AC:BD$ is the relevant bipartition for the state $\rho_{AB}\otimes\sigma_{CD}$.
    \item \textbf{Relation with $k$-unextendible hypothesis testing divergence: } As a straightforward consequence of the following well-known inequality~\cite[Lemma~5]{CMW16}:
    \begin{equation}
        D^{\varepsilon}_H\!\left(\rho\Vert\sigma\right) \le \widetilde{D}_{\alpha}(\rho\Vert\sigma) + \frac{\alpha}{\alpha - 1}\log_2\!\left(\frac{1}{1-\varepsilon}\right), \qquad \forall \alpha \in (1,\infty), \varepsilon \in [0,1)
    \end{equation}
    the following inequality holds for every $k\ge 2$, $\alpha \in (1,\infty)$, and $\varepsilon\in [0,1)$:
    \begin{equation}\label{eq:hypo_test_ent_le_sandwich}
        E^{\varepsilon}_k\!\left(\rho_{AB}\right) \le \widetilde{E}^{\alpha}_k(\rho_{AB}) + \frac{\alpha}{\alpha - 1}\log_2\!\left(\frac{1}{1-\varepsilon}\right).
    \end{equation}
    \item \textbf{Efficiently computable:} The $k$-unextendible sandwiched R\'enyi divergence can be efficiently computed for a fixed $k\ge 2$ and some fixed $\alpha \in \left[\frac{1}{2},1\right)\cup(1,2]$ using the results from~\cite{HSF25}. Furthermore, in the limit $\alpha \to \infty$, the $k$-unextendible sandwiched R\'enyi divergence converges to the $k$-unextendible divergence induced by the max-relative entropy~\cite{Dat09}. This quantity, denoted by $E^{\max}_k$ in~\cite{KDWW19} and~\cite{KDWW21}, can be computed using a semidefinite program, which we detail in Appendix~\ref{app:SDPs}.

\end{enumerate}

\section{Privacy test for \texorpdfstring{$k$}{k}-extendible states}\label{sec:priv_test}

In this section, we establish an upper bound on the probability with which a $k$-extendible state can pass a privacy test. We later use this bound, which we formally state in Theorem~\ref{thm:priv_pass_prob_k_ext}, to obtain limits on the one-shot, one-way distillable key of a state in Section~\ref{sec:one_way_key_distill} and the  one-shot, forward-assisted private capacity of a channel in Section~\ref{sec:one_way_priv_cap}.

Before turning our attention to arbitrary $k$-extendible states, let us first examine a special class of $k$-extendible states, which we call \emph{$k$-pure extendible} states. The notion of \emph{pure extendible} states was introduced in~\cite{ML09}, where they considered two-extendibility of states only. Here we generalize the idea to $k$-extendibility and obtain results that are analogous to~\cite[Lemma 2]{ML09} and~\cite[Corollary 3]{ML09}.

\begin{definition}[$k$-pure extendible state]
    A bipartite state $\rho_{AB}$ is said to be $k$-pure extendible if there exists a pure state $\psi_{AB_{[k]}}$ such that 
    \begin{equation}
        \operatorname{Tr}_{B_{[k]\setminus 1}}\!\left[\psi_{AB_{[k]}}\right] = \rho_{AB},
    \end{equation}
    and
    \begin{equation}
        W^{\pi}_{B_{[k]}}\psi_{AB_{[k]}}\left(W^{\pi}_{B_{[k]}}\right)^{\dagger} = \psi_{AB_{[k]}}\qquad \forall \pi \in S_k,
    \end{equation}
    where $W^{\pi}$ is the permutation operator corresponding to the permutation $\pi$ in the symmetric group~$S_k$.
\end{definition}

\begin{proposition}\label{prop:pure_k_ext_iff}
    A bipartite state is $k$-extendible if and only if it can be written as a convex combination of $k$-pure extendible states.
\end{proposition}
\begin{proof}
    The forward implication is trivial because a convex combination of $k$-extendible states is $k$-extendible. To see the reverse implication, let $\rho_{AB}$ be an arbitrary $k$-extendible state and let $\sigma_{AB_{[k]}}$ be a $k$-extension of $\rho_{AB}$. The permutation invariance condition on $k$-extendible states implies that 
    \begin{equation}
        \left[I_A\otimes W^{\pi}_{B_{[k]}},\sigma_{AB_{[k]}}\right] = 0 \qquad \forall \pi \in S_k.
    \end{equation}
    As such, $I_A\otimes W^{\pi}_{B_{[k]}}$ and $\sigma_{AB_{[k]}}$ share a common normal eigenbasis, say $\left\{|\psi^j\rangle_{AB_{[k]}}\right\}_{j=0}^{|A||B|^k-1}$, for all $\pi \in S_k$. We can then write 
    \begin{equation}
        \sigma_{AB_{[k]}} = \sum_{i=0}^{|A||B|^{k}-1}\lambda_j|\psi^{j}\rangle\!\langle \psi^j|_{AB_{[k]}},
    \end{equation}
    where $\lambda_j$ are probability masses.

    The eigenvalues of a permutation operator are given by the roots of identity. This is evident from the fact that all elements in the symmetric group have a finite order. That is, for every element $\pi \in S_k$, there exists a positive integer $n$ such that $\left(W^{\pi}\right)^n = I$. Since each $|\psi^j\rangle_{AB_{[k]}}$ is an eigenvector of $I_A\otimes W^{\pi}_{B_{[k]}}$,
    \begin{equation}
        W^{\pi}_{B_{[k]}}|\psi^j\rangle_{AB_{[k]}} = \omega_{j,\pi}|\psi^j\rangle_{AB_{[k]}},
    \end{equation}
    where $\omega_{j,\pi}$ is the $n^{\text{th}}$ root of identity for some positive integer $n$. As such,
    \begin{equation}
        W^{\pi}_{B_{[k]}}|\psi^j\rangle\!\langle \psi^j|_{AB_{[k]}}\left( W^{\pi}_{B_{[k]}}\right)^{\dagger} = |\omega_{j,\pi}|^2|\psi^j\rangle\!\langle \psi^j|_{AB_{[k]}} = |\psi^j\rangle\!\langle \psi^j|_{AB_{[k]}}.
    \end{equation}
    Since the last equality holds for every $\pi \in S_k$, we conclude that $\operatorname{Tr}_{B_{[k]\setminus 1}}\!\left[|\psi^i\rangle\!\langle \psi^i|_{AB_{[k]}}\right]$ is a $k$-pure extendible state for every $j \in \{0,1,\ldots,|A||B|^k-1\}$. We can then write the state $\rho_{AB}$ as
    \begin{align}
        \rho_{AB} &= \operatorname{Tr}_{B_{[k]\setminus 1}}\!\left[\sigma_{AB_{[k]}}\right]\\
        &= \sum_{j=0}^{|A||B|^{k}-1}\lambda_j\operatorname{Tr}_{B_{[k]\setminus 1}}\!\left[|\psi^{j}\rangle\!\langle \psi^j|_{AB_{[k]}}\right],
    \end{align}
    which is a convex combination of $k$-pure extendible states.
\end{proof}

\begin{theorem}\label{thm:priv_pass_prob_k_ext}
    Let $\sigma_{ABA'B'}$ be a $k$-extendible state with respect to the partition $AA'\!:\!BB'$, with $|A| = |B| = d$. Let $\left\{\Pi^{\gamma}_{ABA'B'}, I_{ABA'B'} - \Pi^{\gamma}_{ABA'B'}\right\}$ be a privacy test as defined in~\eqref{eq:priv_test_defn}. Then the probability of $\sigma_{ABA'B'}$ passing the privacy test is bounded from above as follows:
    \begin{equation}
        \operatorname{Tr}\!\left[\Pi^{\gamma}_{ABA'B'}\sigma_{ABA'B'}\right] \le \frac{1}{d} + \frac{1}{k} - \frac{1}{dk}.
    \end{equation}
\end{theorem}

\begin{proof}
    We first prove the statement of the theorem for $k$-pure extendible states and then use Proposition~\ref{prop:pure_k_ext_iff} to generalize it to arbitrary $k$-extendible states.
    
    Let $\sigma_{ABA'B'}$ be a $k$-extendible state with respect to the partition $AA'\!:\!BB'$, with $|A| = |B| = d$. Let $\left\{U^i_{A'B'}\right\}_{i=0}^{d-1}$ be a set of unitary operators that determines the privacy test. That is,
    \begin{equation}
        \Pi^{\gamma}_{ABA'B'} = V_{ABA'B'}\!\left(\Phi^d_{AB}\otimes I_{A'B'}\right)V^{\dagger}_{ABA'B'},
    \end{equation}
    where 
    \begin{equation}\label{eq:twisting_unitary_k_ext_proof}
        V_{ABA'B'} \coloneqq \sum_{i=0}^{d-1}|i\rangle\!\langle i|_A\otimes I_B\otimes U^i_{A'B'}.
    \end{equation}
    Note that
    \begin{align}
        \operatorname{Tr}\!\left[\Pi^{\gamma}_{ABA'B'}\sigma_{ABA'B'}\right] &= \operatorname{Tr}\!\left[V_{ABA'B'}\!\left(\Phi^d_{AB}\otimes I_{A'B'}\right)V^{\dagger}_{ABA'B'}\sigma_{ABA'B'}\right]\\
        &= \operatorname{Tr}\!\left[\!\left(\Phi^d_{AB}\otimes I_{A'B'}\right)V^{\dagger}_{ABA'B'}\sigma_{ABA'B'}V_{ABA'B'}\right]\\
        &= \operatorname{Tr}\!\left[\Phi^d_{AB}\operatorname{Tr}_{A'B'}\!\left[V^{\dagger}_{ABA'B'}\sigma_{ABA'B'}V_{ABA'B'}\right]\right]\\
        &= F\!\left(\Phi^d_{AB},\operatorname{Tr}_{A'B'}\!\left[V^{\dagger}_{ABA'B'}\sigma_{ABA'B'}V_{ABA'B'}\right]\right),\label{eq:priv_pass_prob_fid_k}
    \end{align}
    where the second equality follows from the cyclicity of trace and the final equality follows from the fact that the fidelity between a pure state $\psi$ and a mixed state $\sigma$ is equal to $\operatorname{Tr}\!\left[\psi\sigma\right]$. Let $\psi^{\sigma}_{AB_{[k]}A'B'_{[k]}}$ be a $k$-pure extension of $\sigma_{ABA'B'}$. We know from Uhlmann's theorem that there exists a state $\psi^{\tau}_{A'B'_{[k]}B_{[k]\setminus 1}}$ such that
    \begin{equation}\label{eq:uhlmanns_app}
        F\!\left(\Phi^d_{AB},\operatorname{Tr}_{A'B'}\!\left[V^{\dagger}_{ABA'B'}\sigma_{ABA'B'}V_{ABA'B'}\right]\right) = \left|\left(\langle \Phi^d|_{AB}\otimes \langle \psi^{\tau}|_{A'B'_{[k]}B_{[k]\setminus 1}}\right)V^{\dagger}_{ABA'B'}|\psi^{\sigma}\rangle\right|^2.
    \end{equation}
    
    Let $S_{B_1B_2}$ be the swap operator on systems $B_1$ and $B_2$. Now consider the following vector:
    \begin{align}
        |\Psi\rangle &\coloneqq \sum_{i=1}^k \left(S_{B_1B_i}\otimes S_{B'_1B'_i}\right)V_{AB_1A'B'_1}\left(|\Phi^d\rangle_{AB_1}\otimes |\psi^{\tau}\rangle_{A'B'_{[k]}B_{[k]\setminus 1}}\right)\\
        &= \sum_{i=1}^k V_{AB_iA'B'_i}\left(S_{B_1B_i}\otimes S_{B'_1B'_i}\right)\left(|\Phi^d\rangle_{AB_1}\otimes |\psi^{\tau}\rangle_{A'B'_{[k]}B_{[k]\setminus 1}}\right)\\
        &= \sum_{i=1}^k V_{AB_iA'B'_i}\left(|\Phi^d\rangle_{AB_i}\otimes |\varphi^{i}\rangle_{A'B'_{[k]}B_{[k]\setminus i}}\right),
    \end{align}
    where $\varphi^i$ is some normalized state vector, the details of which are not necessary for the proof, and for this reason we have also suppressed the dependence on the symbol $\tau$ in the notation.
    
    Since $|\psi^{\sigma}\rangle$ is a $k$-pure extension of $\sigma_{ABA'B'}$, 
    \begin{equation}
        \left(S_{B_1B_i}\otimes S_{B_1B_i}\right)|\psi^{\sigma}\rangle = \left(S_{B_1B_i}\otimes S_{B'_1B'_i}\right)^{\dagger}|\psi^{\sigma}\rangle = |\psi^{\sigma}\rangle \qquad \forall i\in [k],
    \end{equation}
    where the first equality follows from the fact that the swap operator is self-adjoint. Clearly,
    \begin{align}
        \langle \Psi|\psi^{\sigma}\rangle &= \sum_{i=1}^k\left(\langle \Phi^d|\otimes \langle \psi^{\tau}|\right)V^{\dagger}_{ABA'B'}\left(S_{B_1B_i}\otimes S_{B'_1B'_i}\right)^{\dagger}|\psi^{\sigma}\rangle\\
        &= k\left(\langle \Phi^d|\otimes \langle \psi^{\tau}|\right)V^{\dagger}_{ABA'B'}|\psi^{\sigma}\rangle\label{eq:ext_constr_overlap}
    \end{align}
    Recall from~\eqref{eq:priv_pass_prob_fid_k} and~\eqref{eq:uhlmanns_app} that $|\langle \Psi|\psi^{\sigma}\rangle|^2$ is proportional to the quantity that we wish to bound from above. To achieve this goal we use the Cauchy--Schwarz inequality,
    \begin{equation}\label{eq:CS_inequality}
         |\langle \Psi|\psi^{\sigma}\rangle|^2 \le \langle \psi^{\sigma}|\psi^{\sigma}\rangle \langle \Psi|\Psi\rangle = \langle \Psi|\Psi\rangle,
    \end{equation}
    where the equality follows because $|\psi^{\sigma}\rangle$ is a normalized state vector.

    Let us now evaluate $\langle \Psi|\Psi\rangle$. Consider the following inner product:
    \begin{equation}\label{eq:beta_ij_defn}
        \beta_{ij} \coloneqq \left(\langle\varphi^i|_{A'B'_{[k]}B_{[k]\setminus i}}\otimes\langle \Phi^d|_{AB_i}\right)V^{\dagger}_{AB_iA'B'_i}V_{AB_jA'B'_j}\left(|\Phi^d\rangle_{AB_j}\otimes|\varphi^j\rangle_{A'B'_{[k]}B_{[k]\setminus j}}\right).
    \end{equation}
    It can be easily verified that $\langle \Psi|\Psi\rangle = \sum_{i,j=1}^k\beta_{ij}$ and $\beta_{ii} = 1$ for every $i\in [k]$. We can expand the state $|\varphi^i\rangle$ in the computational basis and write
    \begin{equation}\label{eq:phi_i_SD}
    |\varphi^{i}\rangle_{A'B'_{[k]}B_{[k]\setminus i}} = \sum_{\ell = 0}^{d-1} \lambda^{\ell}_{i,j}|\ell\rangle_{B_j}|\xi^{\ell}_{i,j}\rangle_{A'B'_{[k]}B_{[k]\setminus \{i,j\}}},
    \end{equation}
     where $\left\{|\xi^{\ell}_{i,j}\rangle\right\}_{\ell = 0}^{d-1}$ is a set of normalized state vectors  and $\left\{\lambda^{\ell}_{i,j}\right\}_{\ell=0}^{d-1}$ is a set of complex numbers satisfying the normalization condition $\sum_{\ell=0}^{d-1}\left|\lambda_{i,j}^{\ell}\right|^2 = 1$ for every $i,j\in [k]$. Expanding~\eqref{eq:beta_ij_defn} for $i\neq j$ using~\eqref{eq:phi_i_SD}, we have
    \begin{equation}\label{eq:beta_expanded}
        \beta_{ij}
        = \sum_{\ell,p=0}^{d-1}\overline{\lambda^{\ell}_{i,j}}\lambda^{p}_{j,i}\langle \xi^{\ell}_{i,j}|\langle \ell|_{B_j}\langle \Phi^d|_{AB_i}V^{\dagger}_{AB_iA'B'_i}V_{AB_jA'B'_j}|\Phi^d\rangle_{AB_j}|p\rangle_{B_i}|\xi^{p}_{j,i}\rangle.
    \end{equation}
    Using the definition of the twisting unitary from~\eqref{eq:twist_unitary_defn}, we can write
    \begin{align}
        &V_{AB_jA'B'_j}|\Phi^d\rangle_{AB_j}|p\rangle_{B_i}|\xi^{p}_{j,i}\rangle_{A'B'_{[k]}B_{[k]\setminus \{i,j\}}}\notag\\ &= \sum_{m-0}^{d-1}|m\rangle\!\langle m|_A\otimes I_{B_j}\otimes U^m_{A'B'_j}|\Phi^d\rangle_{AB_j}|p\rangle_{B_i}|\xi^p_{j,i}\rangle_{A'B'_{[k]}B_{[k]\setminus \{i,j\}}}\\
        &= \sum_{m=0}^{d-1}|m\rangle\!\langle m|_A\left(\frac{1}{\sqrt{d}}\sum_{n=0}^{d-1}|n\rangle_A|n\rangle_{B_j}\right)|p\rangle_{B_i}U^m_{A'B'_j}|\xi^p_{j,i}\rangle_{A'B'_{[k]}B_{[k]\setminus \{i,j\}}}\\
        &= \frac{1}{\sqrt{d}}\sum_{m=0}^{d-1}|m\rangle_A|m\rangle_{B_j}|p\rangle_{B_i}U^m_{A'B'_j}|\xi^p_{j,i}\rangle_{A'B'_{[k]}B_{[k]\setminus \{i,j\}}}.
    \end{align}
    Substituting the above equality into~\eqref{eq:beta_expanded}, we can write
    \begin{align}
        \beta_{ij} &= \frac{1}{d}\sum_{m,n,\ell,p=0}^{d-1}\overline{\lambda^{\ell}_{i,j}}\lambda^p_{j,i}\langle \xi^{\ell}_{i,j}|\left(U^n_{A'B'_i}\right)^{\dagger}\langle \ell|_{B_j}\langle n|_{B_i}\langle n|m\rangle_A|m\rangle_{B_j}|p\rangle_{B_i}U^m_{A'B'_j}|\xi^p_{j,i}\rangle\\
        &= \frac{1}{d}\sum_{m,n,\ell,p=0}^{d-1}\overline{\lambda^{\ell}_{i,j}}\lambda^p_{j,i}\langle \xi^{\ell}_{i,j}|\left(U^n_{A'B'_i}\right)^{\dagger}U^m_{A'B'_j}|\xi^p_{j,i}\rangle\delta_{m,n}\delta_{n,p}\delta_{\ell,m}\\
        &= \frac{1}{d}\sum_{m=0}^{d-1}\overline{\lambda^{m}_{i,j}}\lambda^m_{j,i} \langle \xi^m_{i,j}|\left(U^m_{A'B'_i}\right)^{\dagger}U^m_{A'B'_j}|\xi^m_{j,i}\rangle.  
    \end{align}
    Furthermore,
    \begin{align}
        |\beta_{ij}| &\le \frac{1}{d}\sum_{m=0}^{d-1}\overline{\lambda^{m}_{i,j}}\lambda^m_{j,i}\left|\langle \xi^m_{i,j}|\left(U^m_{A'B'_i}\right)^{\dagger}U^m_{A'B'_j}|\xi^m_{j,i}\rangle\right|\\
        &\le \frac{1}{d}\sum_{m=0}^{d-1}\overline{\lambda^{m}_{i,j}}\lambda^m_{j,i},\\
        &\le \frac{1}{d}\left(\sum_{m=0}^{d-1}\left|\lambda^m_{i,j}\right|^2\right)\left(\sum_{m=0}^{d-1}\left|\lambda^m_{j,i}\right|^2\right)\\
        &= \frac{1}{d},\label{eq:beta_ij_ub}
    \end{align}
    where the first inequality follows from the triangle inequality, the second inequality follows from the fact that $U^m|\xi^m\rangle$ is a normalized state vector and the absolute value of its overlap with another state vector is less than or equal to one, the third inequality follows from Cauchy--Schwarz inequality, and the equality follows from the normalization condition of the state vector given in~\eqref{eq:phi_i_SD}. Using the triangle inequality once again, we arrive at the following:
    \begin{align}
        |\langle\Psi|\Psi\rangle| &= \left|\sum_{i,j=0}^{k-1}\beta_{ij}\right|\\
        &\le \sum_{i,j=0}^{k-1}|\beta_{ij}|\\
        &\le \sum_{i=0}^{k-1}|\beta_{ii}| + \sum_{\substack{i,j=0,\\i\neq j}}^{k-1}|\beta_{ij}|\\
        &\le k + \frac{1}{d}(k^2-k),
    \end{align}
    where the final inequality follows from~\eqref{eq:beta_ij_ub} and the fact that $\beta_{ii} = 1$ for every $i\in \{0,1,\ldots,k-1\}$.  Now using~\eqref{eq:CS_inequality} and~\eqref{eq:ext_constr_overlap}, we have
    \begin{equation}
        k^2|\langle\Phi^d\otimes \psi^{\tau}|V^{\dagger}_{ABA'B'}|\psi^{\sigma}\rangle|^2 \le \langle\Psi|\Psi\rangle\\
        \le k\left(1 + \frac{k-1}{d}\right).
    \end{equation}
    Substituting the above inequality into~\eqref{eq:uhlmanns_app} and using~\eqref{eq:priv_pass_prob_fid_k}, we arrive at the following inequality:
    \begin{equation}
        \operatorname{Tr}\!\left[\Pi^{\gamma}_{ABA'B'}\sigma_{ABA'B'}\right] \le \frac{1}{d} + \frac{1}{k} - \frac{1}{dk},
    \end{equation}
    which holds for every $k$-pure extendible state $\sigma_{ABA'B'}$ and every privacy test $\left\{\Pi^{\gamma}, I - \Pi^{\gamma}\right\}$ when the dimension of each key system is equal to $d$. Finally, since every $k$-extendible state can be written as a convex combination of $k$-pure extendible states, as stated in Proposition~\ref{prop:pure_k_ext_iff}, we conclude that the statement of the theorem holds for every $k$-extendible state.
\end{proof}

\section{Limits on one-way secret-key distillation from states}\label{sec:one_way_key_distill}

In this section, we obtain an SDP computable upper bound on the number of secret bits that can be distilled from an arbitrary bipartite state in the one-shot regime  using one-way LOCC channels.

Let us begin by defining the quantity of interest, which is the one-shot, one-way distillable key of a bipartite state. There are several ways to quantify the error in a key distillation protocol, which leads to different definitions of the one-shot, one-way distillable key of a state (see for example ~\cite{RR12,KKGW21}). In this work, we use the error criterion from~\cite{WTB17}.

\begin{definition}\label{def:dd_1shot_key_defn}
    The one-shot, one-way distillable key of a state is defined as follows:
    \begin{equation}\label{eq:dd_1shot_key_defn}
        K^{\varepsilon,\to}\!\left(\rho_{AB}\right) \coloneqq \sup_{\substack{d\in \mathbb
        N,\\ \gamma^d_{A'B'A''B''},\\ \mathcal{L}^{\to }\in \operatorname{1WL}}}\left\{\log_2d: F\!\left(\mathcal{L}^{\to}_{AB\to A'B'A''B''}\!\left(\rho_{AB}\right), \gamma^d_{A'B'A''B''}\right)\ge 1-\varepsilon\right\},
    \end{equation}
    where the supremum is over every $d\in \mathbb{N}$, private state $\gamma^d_{A'B'A''B''}$, and one-way LOCC channel $\mathcal{L}^{\to}_{AB\to A'B'A''B''}$.
\end{definition}

\subsection{Upper bounds on the one-shot, one-way distillable key of a state}

In this section, we obtain an upper bound on the one-shot, one-way distillable key of a state using the $k$-unextendible hypothesis testing divergence. 

\begin{theorem}\label{theo:dd_1shot_key_bnd_k_ext}
    Fix $k\ge 2$ and $\varepsilon \in [0,1]$. If $E^{\varepsilon}_{k}\!\left(\rho_{AB}\right) \le \log_2k$, then the one-shot, one-way distillable key of a state $\rho_{AB}$ is bounded from above by the following quantity:
    \begin{equation}\label{eq:dd_1shot_key_ub_k_ext}
        K^{\varepsilon,\to}\!\left(\rho_{AB}\right) \le -\log_2\!\left(2^{-E^{\varepsilon}_{k}\!\left(\rho_{AB}\right)} - \frac{1}{k}\right) + \log_2\!\left(\frac{k-1}{k}\right).
    \end{equation}
\end{theorem}
\begin{proof}
    Let $\mathcal{L}^{\to}_{AB\to A'B'A''B''}$ be a one-way LOCC channel, and let $\gamma^d_{A'B'A''B''}$ be a private state such that
    \begin{equation}
        F\!\left(\mathcal{L}^{\to}\!\left(\rho_{AB}\right), \gamma^d_{A'B'A''B''}\right) \ge 1-\varepsilon.
    \end{equation}
    Then we know from~\eqref{eq:priv_test_pass_prob} that
    \begin{equation}
        \operatorname{Tr}\!\left[\Pi^{\gamma}_{A'B'A''B''}\mathcal{L}^{\to}(\rho_{AB})\right] \ge 1-\varepsilon.
    \end{equation}

    Let $\sigma_{AB}$ be a $k$-extendible state. Then $\mathcal{L}^{\to}_{AB\to A'B'A''B''}(\sigma_{AB})$ is also a $k$-extendible state since a one-way LOCC channel preserves the $k$-extendibility of a state. Consequently, 
    \begin{equation}
        \operatorname{Tr}\!\left[\Pi^{\gamma}_{A'B'A''B''}\mathcal{L}^{\to}(\sigma_{AB})\right] \le\frac{1}{d} +\frac{1}{k} - \frac{1}{dk} 
    \end{equation}
    as per Theorem~\ref{thm:priv_pass_prob_k_ext}.
    
    Recall the definition of hypothesis testing relative entropy from~\eqref{eq:hypo_test_defn}. Since $ \Pi^{\gamma}_{A'B'A''B''}$ is a valid measurement operator, it follows that
    \begin{equation}
        D^{\varepsilon}_H\!\left(\mathcal{L}^{\to}\!\left(\rho_{AB}\right)\middle \Vert\mathcal{L}^{\to}\!\left(\sigma_{AB}\right)\right) \ge -\log_2\!\left(\frac{1}{d} + \frac{1}{k} - \frac{1}{dk}\right).
    \end{equation}
    Furthermore, the data-processing inequality for the hypothesis testing relative entropy implies that
    \begin{equation}
        D^{\varepsilon}_H\!\left(\rho_{AB}\middle \Vert\sigma_{AB}\right)\ge D^{\varepsilon}_H\!\left(\mathcal{L}^{\to}\!\left(\rho_{AB}\right)\middle \Vert\mathcal{L}^{\to}\!\left(\sigma_{AB}\right)\right) \ge -\log_2\!\left(\frac{1}{d} + \frac{1}{k} - \frac{1}{dk}\right).
    \end{equation}
    Since the above inequality holds for every $k$-extendible state $\sigma_{AB}$, we conclude that
    \begin{equation}
        E^{\varepsilon}_{k}\!\left(\rho_{AB}\right) \ge -\log_2\!\left(\frac{1}{d} + \frac{1}{k} - \frac{1}{dk}\right),
    \end{equation}
    where $E^{\varepsilon}_k(\cdot)$ is defined in~\eqref{eq:unext_ent_hypo_test}. If $E^{\varepsilon}_{k}\!\left(\rho_{AB}\right) \le \log_2k$, we can rearrange the above inequality to arrive at the following (see~\cite[Remark~3]{KDWW21} for more mathematical details):
    \begin{equation}
        \log_2 d \le  \log_2\!\left(\frac{k-1}{k}\right)- \log_2\!\left(2^{-E^{\varepsilon}_{k}\!\left(\rho_{AB}\right)}-\frac{1}{k}\right).
    \end{equation}
    Note that the above inequality holds for every positive integer $d$ such that $F\!\left(\mathcal{L}^{\to}(\rho_{AB}),\gamma^d_{A'B'A''B''}\right)\ge 1-\varepsilon$ for some private state $\gamma^d_{A'B'A''B''}$ and some one-way LOCC channel $\mathcal{L}^{\to}_{AB\to A'B'A''B''}$. Therefore, by definition of the one-shot, one-way distillable key of a state,
    \begin{equation}\label{eq:dd_key_ub_sdp_computable}
        K^{\varepsilon,\to}\!\left(\rho_{AB}\right) \le \log_2\!\left(\frac{k-1}{k}\right)- \log_2\!\left(2^{-E^{\varepsilon}_{k}\!\left(\rho_{AB}\right)}-\frac{1}{k}\right).
    \end{equation}
    This concludes the proof.
\end{proof}

\medskip

Recall that the set of $k$-extendible states converges to the set of separable states as $k\to \infty$. As such,
\begin{equation}
    \lim_{k\to \infty}E^{\varepsilon}_k\!\left(\rho_{AB}\right) = E^{\varepsilon}_R\!\left(\rho_{AB}\right),
\end{equation}
where $E^{\varepsilon}_R(\rho)$ is the hypothesis-testing relative entropy of entanglement of the state $\rho$ defined as~\cite{BD11}:
\begin{equation}
    E^{\varepsilon}_R\!\left(\rho_{AB}\right) \coloneqq \inf_{\sigma \in \operatorname{SEP}(A:B)}D^{\varepsilon}_H\!\left(\rho_{AB}\Vert\sigma_{AB}\right)
\end{equation}
with $\operatorname{SEP}(A\!:\!B)$ being the set of states that are separable across the bipartition $A\!:\!B$. 

One can verify that, in the limit $k\to \infty$, Theorem~\ref{theo:dd_1shot_key_bnd_k_ext} implies that
\begin{equation}\label{eq:hypo_test_ent_key_ub}
    K^{\varepsilon,\to}\!\left(\rho_{AB}\right) \le E^{\varepsilon}_R\!\left(\rho_{AB}\right).
\end{equation}
The hypothesis-testing relative entropy of entanglement of a state is known to be an upper bound on the one-shot distillable key of the state, even when both Alice and Bob can communicate with each other publicly~\cite{WTB17}. Naturally, the upper bound on the one-shot distillable key of a state from~\cite{WTB17} implies that the hypothesis-testing relative entropy of entanglement is also an upper bound on the one-shot, one-way distillable key of the state, which is what we have recovered in~\eqref{eq:hypo_test_ent_key_ub}.

\subsection{Upper bounds on the \texorpdfstring{$n$}{n}-shot, one-way distillable key of a state}

We are often interested in distilling secret keys from multiple copies of a state, which is called the $n$-shot, one-way distillable key of the state, with $n$ equal to the number of copies of the state  used in the distillation protocol.  In principle, one can compute the $n$-shot one-way distillable key of a state $\rho_{AB}$ by simply computing the bound in Theorem~\ref{theo:dd_1shot_key_bnd_k_ext} for  $\rho_{AB}^{\otimes n}$, but the time complexity of the semidefinite program (SDP) scales exponentially with $n$, making such a computation impractical. One can possibly reduce the time complexity of this computation to $O(\text{poly}(n))$ by following the approach of~\cite{FST22}. Such a reduction was found in~\cite{SNW25} specifically in the context of $k$-extendibility. 

Here we  relax the bound in Theorem~\ref{theo:dd_1shot_key_bnd_k_ext} to obtain a single-letter upper bound on the $n$-shot, one-way distillable key of a state by considering the $k$-unextendible sandwiched R\'enyi divergence. 

Recall from~\eqref{eq:hypo_test_ent_le_sandwich} that
\begin{equation}
    E^{\varepsilon}_k\!\left(\rho_{AB}\right) \le \widetilde{E}^{\alpha}_k\!\left(\rho_{AB}\right) + \frac{\alpha}{\alpha - 1}\log_2\!\left(\frac{1}{1-\varepsilon}\right) \qquad \forall \alpha \in (1,\infty), k\ge 2, \varepsilon\in [0,1),
\end{equation}
where $\widetilde{E}^{\alpha}_k$ is the $k$-extendible sandwiched R\'enyi divergence defined in~\eqref{eq:sandwich_unext_ent}. 
The subadditivity of $\widetilde{E}^{\alpha}_k$ under  tensor products implies that
\begin{equation}
    E^{\varepsilon}_k\!\left(\rho^{\otimes n}_{AB}\right) \le n\widetilde{E}^{\alpha}_k\!\left(\rho_{AB}\right) + \frac{\alpha}{\alpha - 1}\log_2\!\left(\frac{1}{1-\varepsilon}\right) \qquad \forall \alpha \in (1,\infty), k\ge 2, \varepsilon\in [0,1), n\in \mathbb{N}.
\end{equation}
Substituting the above inequality into~\eqref{eq:dd_1shot_key_ub_k_ext}, we arrive at a single-letter upper bound on the $n$-shot, one-way distillable key of a state, which we state formally in Corollary~\ref{cor:DD_key_n_shot_ub_sandwich} below.
\begin{corollary}\label{cor:DD_key_n_shot_ub_sandwich}
    Fix $\varepsilon \in [0,1)$, $\alpha \in (1,\infty)$, and an integer $k\ge 2$. Let $\rho_{AB}$ be an arbitrary bipartite state. If 
    \begin{equation}
    n\widetilde{E}^{\alpha}_k\!\left(\rho_{AB}\right) + \frac{\alpha}{\alpha - 1}\log_2\!\left(\frac{1}{1-\varepsilon}\right) \le \log_2k,
\end{equation}
    then the one-shot, one-way distillable key of a state is bounded from above as follows: \begin{equation}\label{eq:DD_key_n_shot_ub_sandwich}
    K^{\varepsilon,\to}\left(\rho^{\otimes n}_{AB}\right) \le \log_2\!\left(\frac{k-1}{k}\right) - \log_2\!\left(2^{-n\widetilde{E}^{\alpha}_k(\rho)}(1-\varepsilon)^{\frac{\alpha}{\alpha - 1}}-\frac{1}{k}\right).
\end{equation}
\end{corollary}

As mentioned earlier in Section~\ref{sec:sandwich_k_unext_st}, the $k$-unextendible sandwiched R\'enyi divergence of states can be efficiently computed for $\alpha \in \left[\frac{1}{2},1\right)\cup(1,2]\cup\{+\infty\}$. Therefore,~\eqref{eq:DD_key_n_shot_ub_sandwich} yields a single-letter, efficiently computable upper bound on the $n$-shot, one-way distillable key of a state.

\section{Limits on forward-assisted private capacity of channels}\label{sec:one_way_priv_cap}

In this section, we analyze limits on private communication over a channel assisted by forward classical communication. 

Let us first consider the task of secret-key distillation from a channel, also known as secret-key generation~\cite{WTB17}. In this task, Alice and Bob wish to establish a secret key between them using a quantum channel $\mathcal{N}_{A\to B}$. To achieve this, Alice prepares an arbitrary bipartite state and sends one share of the state to Bob using the channel $\mathcal{N}$. Bob then applies a quantum channel to the state he received, which we shall call the decoding channel, such that the bipartite state established between Alice and Bob at the end of this protocol is close to a private state. 

The ability to establish a secret key using a channel $\mathcal{N}_{A\to B}$ can be quantified by the one-shot distillable key of the channel, which is defined for an error parameter $\varepsilon\in [0,1]$ as follows (see~\cite{WTB17}):
\begin{equation}\label{eq:1shot_dist_key_ch_defn}
    K^{\varepsilon}\!\left(\mathcal{N}_{A\to B}\right) \coloneqq \sup_{\substack{\rho_{AA'A''}\in \mathcal{S}(AA'A''),\\ \mathcal{D}_{B\to B'B''}\in \operatorname{CPTP}, \\ d\in \mathbb{N}, ~\gamma^d_{A'A''B'B''}}}\left\{\log_2d: F\!\left(\mathcal{D}_{B\to B'B''}\circ\mathcal{N}_{A\to B}\!\left(\rho_{AA'A''}\right), \gamma^d_{A'A''B'B''}\right)\ge 1-\varepsilon\right\},
\end{equation}
where the supremum is over every positive integer $d$, every state $\rho_{AA'A''}$, every channel $\mathcal{D}_{B\to B'B''}$, and every private state $\gamma^d_{A'A''B'B''}$ with $d = |A'| = |B'|$. Comparing with Definition~\ref{def:dd_1shot_key_defn}, it can be easily seen that
\begin{equation}
    K^{\varepsilon}\!\left(\mathcal{N}_{A\to B}\right) \le \sup_{\rho_{AA'A''}\in \mathcal{S}(AA'A'')}K^{\varepsilon}\!\left(\mathcal{N}_{A\to B}\!\left(\rho_{AA'A''}\right)\right)
\end{equation}
since the local channel $\mathcal{D}_{B\to B'B''}$ is an instance of a one-way LOCC channel. We note that that systems $A'$ and $A''$ can be arbitrarily large in the above inequality.

Now consider the setting where Alice can publicly announce an arbitrary amount of classical data beside the channel $\mathcal{N}$. The quantity of interest in this setting is the one-shot, one-way distillable key of a channel, which is defined for an error parameter $\varepsilon \in [0,1]$ as follows:
\begin{multline}\label{eq:1shot_1w_key_ch_defn}
    K^{\varepsilon,\to}\!\left(\mathcal{N}_{A\to B}\right)\\ \coloneqq \sup_{\substack{\rho_{XAA'A''}\in \mathcal{S}(XAA'A''),\\ \mathcal{D}_{BX\to B'B''}\in \operatorname{CPTP}, \\ d\in \mathbb{N}, ~\gamma^d_{A'A''B'B''}}}\left\{\log_2d: F\!\left(\mathcal{D}_{BX\to B'B''}\circ\mathcal{N}_{A\to B}\!\left(\rho_{XAA'A''}\right), \gamma^d_{A'A''B'B''}\right)\ge 1-\varepsilon\right\},
\end{multline}
where system $X$ is classical and the remaining symbols have the same meaning as in~\eqref{eq:1shot_dist_key_ch_defn}. Once again, it can be easily verified that
\begin{equation}\label{eq:dist_key_ch_le_dist_key_st}
    K^{\varepsilon,\to}\!\left(\mathcal{N}_{A\to B}\right) \le \sup_{\rho_{AA'A''}\in \mathcal{S}(AA'A'')}K^{\varepsilon,\to}\!\left(\mathcal{N}_{A\to B}\!\left(\rho_{AA'A''}\right)\right).
\end{equation}

In general, there may exist protocols that facilitate secure communication over a channel without explicitly establishing a secret key~\cite{DLL03}. The notion of private capacity of a channel~\cite{Dev05, CWY04} is then more useful to quantify the amount of data that can be securely transmitted over the channel.

We follow~\cite[Chapter 16]{KW20} to define the one-shot private capacity of a channel, which is consistent with the definition of private capacity of the channel in the asymptotic regime defined in~\cite{Dev05,CWY04}. Let $\mathcal{E}_{X\to A}$ be a channel that Alice uses to encode some classical data on system $X$, drawn with respect to some probability distribution over a symbol $\mathcal{X}$, into a quantum state on system $A$. She then sends the quantum state over the channel $\mathcal{N}_{A\to B}$ to Bob. Bob then uses a decoding channel $\mathcal{D}_{B\to \hat{X}}$ to decode the classical message. The eavesdropper may have access to the purifying system $E$ coming out from some isometric extension $\mathcal{U}^{\mathcal{N}}_{A\to BE}$ of the channel $\mathcal{N}_{A\to B}$. The worst-case error in secure transmission of classical data using this protocol is defined as follows:
\begin{equation}
    p_{\operatorname{err}}(\mathcal{X},\mathcal{E},\mathcal{N},\mathcal{D}) \coloneqq \inf_{\sigma_E}\sup_{x\in \mathcal{X}}\left(1-F\!\left(|x\rangle\!\langle x|_{\hat{X}}\otimes \sigma_E,\mathcal{D}_{B\to \hat{X}}\circ\mathcal{U}^{\mathcal{N}}_{A\to BE}\circ\mathcal{E}_{X\to A}\!\left(|x\rangle\!\langle x|_X\right)\right)\right),
\end{equation}
where the infimum is over every state $\sigma_E$ and the supremum is over every letter $x$ in the alphabet~$\mathcal{X}$. The one-shot private capacity of a channel is then defined as follows:
\begin{equation}
    P^{\varepsilon}\!\left(\mathcal{N}_{A\to B}\right) \coloneqq \sup_{\substack{\mathcal{X}, \\ \mathcal{E}_{X\to A},\mathcal{D}_{B\to \hat{X}}\in \operatorname{CPTP}}} \left\{\log_2|\mathcal{X}|: p_{\operatorname{err}}\!\left(\mathcal{X},\mathcal{E},\mathcal{N},\mathcal{D}\right)\le \varepsilon\right\},
\end{equation}
where the supremum is over every classical alphabet $\mathcal{X}$, every classical-to-quantum channel $\mathcal{E}_{X\to A}$, and every quantum-to-classical channel $\mathcal{D}_{B\to \hat{X}}$.

In the presence of a free classical side channel from Alice to Bob, the quantity of interest is the one-shot, forward-assisted private capacity of the channel, which is defined in the same way as the one-shot forward-assisted private capacity of the channel, but with Alice having the ability to publicly announce an arbitrarily large amount of classical data. We denote the one-shot, forward-assisted private capacity of the channel $\mathcal{N}_{A\to B}$ by the symbol $P^{\varepsilon,\to}\!\left(\mathcal{N}_{A\to B}\right)$.

One can always transform a private communication protocol into a secret-key distillation protocol by transmitting a symbol chosen from a uniform probability distribution. Therefore,
\begin{equation}
    P^{\varepsilon}\!\left(\mathcal{N}_{A\to B}\right) \leq K^{\varepsilon}\!\left(\mathcal{N}_{A\to B}\right).
\end{equation}
However, when forward classical communication can be performed for free, a secret-key distillation protocol can be transformed back into a private communication protocol by using the one-time-pad scheme. Therefore,
\begin{equation}
    P^{\varepsilon,\to}\!\left(\mathcal{N}_{A\to B}\right) = K^{\varepsilon,\to}\!\left(\mathcal{N}_{A\to B}\right).
\end{equation}
Now using~\eqref{eq:dist_key_ch_le_dist_key_st}, we find that
\begin{equation}
    P^{\varepsilon,\to}\!\left(\mathcal{N}_{A\to B}\right) \le \sup_{\rho_{AA'A''}\in \mathcal{S}(AA'A'')}K^{\varepsilon,\to}\!\left(\mathcal{N}_{A\to B}\!\left(\rho_{AA'A''}\right)\right),
\end{equation}
where the dimension of $A'$ and $A''$ can be unbounded. Put differently,
\begin{equation}\label{eq:priv_cap_le_key_st}
    P^{\varepsilon,\to}\!\left(\mathcal{N}_{A\to B}\right) \le \sup_{\rho_{RA}\in \mathcal{S}(RA)}K^{\varepsilon,\to}\!\left(\mathcal{N}_{A\to B}\!\left(\rho_{RA}\right)\right),
\end{equation}
where there is no restriction on the dimension of the system $R$.

\subsection{Upper bounds on the one-shot, forward-assisted private capacity of a channel}

In this section, we obtain semidefinite computable upper bounds on the one-shot, forward-assisted private capacity of a channel by using the inequality in~\eqref{eq:priv_cap_le_key_st} along with the results from Section~\ref{sec:one_way_key_distill}. To achieve this goal, we define the $k$-unextendible generalized divergence of a channel as a measure for quantifying the unextendibility of a point-to-point channel.

\begin{definition}\label{def:k_unext_gen_div_ch_defn}
    The $k$-unextendible generalized divergence of a channel is defined as follows:
\begin{equation}\label{eq:k_unext_gen_div_ch_defn}
    \mathbf{E}_k\!\left(\mathcal{N}_{A\to B}\right) \coloneqq \inf_{\mathcal{M}_{A\to B}\in \operatorname{CPTP}} \sup_{\rho_{RA}\in \mathcal{S}(RA)}\left\{\begin{array}{c}
        \mathbf{D}\!\left(\mathcal{N}_{A\to B}\!\left(\rho_{RA}\right)\middle\Vert\mathcal{M}_{A\to B}\!\left(\rho_{RA}\right)\right):\\ \mathcal{M}_{A\to B}\!\left(\Phi_{A'A}\right)\in \operatorname{Ext}_k(A'\!:\!B) 
        
    \end{array}\right\},
\end{equation}
where $A'\cong A$.
\end{definition}

The quantity 
\begin{equation}
	 \sup_{\rho_{RA}\in \mathcal{S}(RA)}\mathbf{D}\!\left(\mathcal{N}_{A\to B}\!\left(\rho_{RA}\right)\middle\Vert\mathcal{M}_{A\to B}\!\left(\rho_{RA}\right)\right) \eqcolon \mathbf{D}\!\left(\mathcal{N}_{A\to B}\middle\Vert\mathcal{M}_{A\to B}\right).
\end{equation}
is called the generalized divergence of channels~\cite{CMW16,LKDW18}. This allows us to rewrite the $k$-unextendible generalized divergence of a channel more concisely as follows:
\begin{equation}
\mathbf{E}_k\!\left(\mathcal{N}_{A\to B}\right) \coloneqq \inf_{\mathcal{M}_{A\to B}\in \operatorname{CPTP}} \left\{\begin{array}{c}
        \mathbf{D}\!\left(\mathcal{N}_{A\to B}\middle\Vert\mathcal{M}_{A\to B}\right): \mathcal{M}_{A\to B}\!\left(\Phi_{A'A}\right)\in \operatorname{Ext}_k(A'\!:\!B) 
        
    \end{array}\right\},
    \label{eq:k-unext-div-gen-def}
\end{equation}
where $A'\cong A$. 

The state $\Gamma^{\mathcal{M}}_{A'B}\coloneqq \mathcal{M}_{A\to B}\!\left(\Phi_{A'A}\right)$ is called the Choi state of the channel $\mathcal{M}_{A\to B}$. A point-to-point channel whose Choi state is $k$-extendible is called a point-to-point $k$-extendible channel~\cite{PBHS11}. This definition of point-to-point $k$-extendible channels is consistent with the definition of bipartite $k$-extendible channels in the sense that for every $k$-extendible channel $\mathcal{N}_{A\to B}$ there exists an extended channel $\mathcal{P}_{A\to B_{[k]}}$ such that the conditions in~\eqref{eq:marginal_cond_k_ext} and~\eqref{eq:perm_cov_cond_k_ext} are satisfied after fixing $A', B_1,B_2,\ldots,B_k$ to be trivial systems.

To be precise, if the Choi state of a point-to-point channel $\mathcal{M}_{A\to B}$ is $k$-extendible, then there exists a channel $\mathcal{P}_{A\to B_{[k]}}$ such that the following equalities hold:
\begin{equation}
	 \operatorname{Tr}_{B_{[k]\setminus 1}}\circ\mathcal{P}_{A\to B_{[k]}} = \mathcal{M}_{A\to B_1},
\end{equation}
and
\begin{equation}
	\mathcal{W}^{\pi}_{B_{[k]}}\circ\mathcal{P}_{A\to B_{[k]}} = \mathcal{P}_{A\to B_{[k]}}\qquad \forall \pi \in S_k,
\end{equation}
where $\mathcal{W}^{\pi}$ is the permutation channel corresponding to the permutation $\pi$ in the symmetric group~$S_k$.

Note that the equivalence between the $k$-extendibility of a channel and the $k$-extendibility of its Choi state holds only in the case of point-to-point channels. In the bipartite case, the Choi state of a channel being $k$-extendible is necessary for the channel to be $k$-extendible but not sufficient.

We now define the following special cases of \eqref{eq:k-unext-div-gen-def}, which are useful for our purposes in what follows:
\begin{align}
	E^{\varepsilon}_k\!\left(\mathcal{N}_{A\to B}\right) &\coloneqq \inf_{\mathcal{M}_{A\to B}\in \operatorname{CPTP}}\left\{D^{\varepsilon}_H\!\left(\mathcal{N}_{A\to B}\middle\Vert\mathcal{M}_{A\to B}\right): \mathcal{M}_{A\to B}\!\left(\Phi_{RA}\right)\in \operatorname{Ext}_k\!\left(R\!:\!B\right)\right\},\label{eq:hypo_test_ent_ch_defn}\\
    \widetilde{E}^{\alpha}_k\!\left(\mathcal{N}_{A\to B}\right) &\coloneqq \inf_{\mathcal{M}_{A\to B}\in \operatorname{CPTP}}\left\{\widetilde{D}_{\alpha}\!\left(\mathcal{N}\middle\Vert\mathcal{M}\right): \mathcal{M}_{A\to B}\!\left(\Phi_{RA}\right)\in \operatorname{Ext}_k\!\left(R\!:\!B\right)\right\} \quad \forall \alpha \in \left[\frac{1}{2},1\right)\cup(1,\infty),\label{eq:sandwich_ent_ch_defn}\\
	\widehat{E}^{\alpha}_k\!\left(\mathcal{N}_{A\to B}\right) &\coloneqq \inf_{\mathcal{M}_{A\to B}\in \operatorname{CPTP}}\left\{\widehat{D}_{\alpha}\!\left(\mathcal{N}\middle\Vert\mathcal{M}\right): \mathcal{M}_{A\to B}\!\left(\Phi_{RA}\right)\in \operatorname{Ext}_k\!\left(R\!:\!B\right)\right\} \quad \forall \alpha \in (0,1)\cup(1,2],\label{eq:geo_ent_ch_defn}
\end{align}
where $R\cong A$ in all the above equalities. As is evident from their definitions, we call the quantities in~\eqref{eq:hypo_test_ent_ch_defn},~\eqref{eq:sandwich_ent_ch_defn}, and~\eqref{eq:geo_ent_ch_defn} the $k$-unextendible hypothesis testing divergence of a channel, the $k$-unextendible sandwiched R\'enyi divergence of a channel, and the $k$-unextendible geometric R\'enyi divergence of a channel, respectively.

\begin{remark}\label{rem:unext_gen_ch_div_diff}
Our definition of the $k$-unextendible generalized divergence of a channel differs from the one introduced in~\cite{KDWW19,KDWW21}. All the upper bounds on the one-shot, forward-assisted private capacity of a channel obtained in this work (Theorem~\ref{theo:priv_cap_ub_hypo_test} and Corollary~\ref{cor:n_shot_priv_cap_ub}) hold true for both definitions of $k$-unextendible generalized divergence, the one given in Definition~\ref{def:k_unext_gen_div_ch_defn} as well as the one considered in~\cite{KDWW19, KDWW21}. However, it is not clear if the $k$-unextendible generalized divergence of channels defined in~\cite{KDWW19,KDWW21} can be efficiently computed. On the other hand, the quantities defined in~\eqref{eq:hypo_test_ent_ch_defn}--\eqref{eq:geo_ent_ch_defn} can be computed via  semidefinite programs (see Appendix~\ref{app:SDPs} for details), which warrants their use in the rest of this work.
\end{remark}

\begin{lemma}\label{lem:unext_div_ch_ge_unext_div_st}
	For every quantum channel $\mathcal{N}_{A\to B}$ and every integer $k\ge 2$,
	\begin{equation}\label{eq:k_unext_ch_ge_st}
		\mathbf{E}_k\!\left(\mathcal{N}_{A\to B}\right) \ge \sup_{\rho_{RA}\in \mathcal{S}(RA)}\mathbf{E}_k\!\left(\mathcal{N}_{A\to B}\!\left(\rho_{RA}\right)\right).
	\end{equation}
\end{lemma}

\begin{proof}
	Let $\mathcal{M}_{A\to B}$ be a $k$-extendible channel for some integer $k\ge 2$. Then there exists a channel $\mathcal{P}_{A\to B_{[k]}}$ such that 
	\begin{align}
		\operatorname{Tr}_{B_{[k]\setminus 1}}\circ\mathcal{P}_{A\to B_{[k]}} &= \mathcal{M}_{A\to B_1},\\
		\mathcal{W}^{\pi}_{B_{[k]}}\circ\mathcal{P}_{A\to B_{[k]}} &= \mathcal{P}_{A\to B_{[k]}}.
	\end{align}
	Let $\rho_{RA}$ be an arbitrary bipartite state, with system $R$ being of arbitrary dimension. The following equalities hold for the state $\mathcal{P}_{A\to B_{[k]}}\!\left(\rho_{RA}\right)$:
	\begin{align}
		\operatorname{Tr}_{B_{[k]\setminus 1}}\!\left[\mathcal{P}_{A\to B_{[k]}}\!\left(\rho_{RA}\right)\right] &= \mathcal{M}_{A\to B_1}\!\left(\rho_{RA}\right),\\
		\mathcal{W}^{\pi}_{B_{[k]}}\!\left(\mathcal{P}_{A\to B_{[k]}}\!\left(\rho_{RA}\right)\right) &= \mathcal{P}_{A\to B_{[k]}}\!\left(\rho_{RA}\right) \qquad \forall \pi \in S_k.
	\end{align}
	Therefore, $\mathcal{M}_{A\to B}\!\left(\rho_{RA}\right)$ is a $k$-extendible state with $\mathcal{P}_{A\to B_{[k]}}\!\left(\rho_{RA}\right)$ being its $k$-extension. 
	
	Now applying the max-min inequality to~\eqref{eq:k_unext_gen_div_ch_defn}, we arrive at the following inequality:
	\begin{align}
		\mathbf{E}_k\!\left(\mathcal{N}_{A\to B}\right) &\ge  \sup_{\rho_{RA}\in \mathcal{S}(RA)}\inf_{\mathcal{M}_{A\to B}\in \operatorname{CPTP}}\left\{\begin{array}{c}
        \mathbf{D}\!\left(\mathcal{N}_{A\to B}\!\left(\rho_{RA}\right)\middle\Vert\mathcal{M}_{A\to B}\!\left(\rho_{RA}\right)\right):\\ \mathcal{M}_{A\to B}\!\left(\Phi_{RA}\right)\in \operatorname{Ext}_k(R\!:\!B)         
    \end{array}\right\}\\
    &\ge \sup_{\rho_{RA}\in \mathcal{S}(RA)}\inf_{\sigma_{RB}\in \mathcal{S}(RB)}\left\{\begin{array}{c}
        \mathbf{D}\!\left(\mathcal{N}_{A\to B}\!\left(\rho_{RA}\right)\middle\Vert\sigma_{RB}\right):\\ \sigma_{RB}\in \operatorname{Ext}_k(R\!:\!B) 
        \end{array}\right\}\\
        &= \sup_{\rho_{RA}\in \mathcal{S}(RA)}\mathbf{E}_k\!\left(\mathcal{N}_{A\to B}\!\left(\rho_{RA}\right)\right),
	\end{align}
	where the second inequality follows from the fact that $\mathcal{M}_{A\to B}\!\left(\rho_{RA}\right)\in \operatorname{Ext}_k\!\left(R\!:\!B\right)$ if $\mathcal{M}_{A\to B}$ is a point-to-point $k$-extendible channel.
\end{proof}

\medskip

Note that the upper bound on the one-shot, one-way distillable key of a state, stated in Theorem~\ref{theo:dd_1shot_key_bnd_k_ext}, is a monotonically increasing function with respect to the $k$-unextendible hypothesis-testing divergence of the state. We use this fact, along with the inequalities in~\eqref{eq:priv_cap_le_key_st} and~\eqref{eq:k_unext_ch_ge_st}, to obtain an upper bound on the one-shot, forward-assisted private capacity of a channel, which we state formally in Theorem~\ref{theo:priv_cap_ub_hypo_test} below.
\begin{theorem}\label{theo:priv_cap_ub_hypo_test}
    The one-shot, forward-assisted private capacity of a channel $\mathcal{N}_{A\to B}$ is bounded from above by the following quantity:
    \begin{equation}
        P^{\varepsilon,\to}\!\left(\mathcal{N}_{A\to B}\right) \le \log_2\!\left(\frac{k-1}{k}\right) -\log_2\!\left(2^{-E^{\varepsilon}_k(\mathcal{N})} - \frac{1}{k}\right),
    \end{equation}
    where $E^{\varepsilon}_k\!\left(\mathcal{N}\right)$ is defined in~\eqref{eq:hypo_test_ent_ch_defn}.
\end{theorem}
\begin{proof}
	See Appendix~\ref{app:priv_cap_ub}.    
\end{proof}

\medskip

\subsection{Upper bounds on the \texorpdfstring{$n$}{n}-shot, forward-assisted private capacity}

In this section, we obtain single-letter upper bounds on the $n$-shot, forward-assisted private capacity of a channel.

First, we note that the $\alpha$-geometric R\'enyi relative entropy of channels is additive under tensor products for every $\alpha \in \left(\frac{1}{2},1\right)\cup(1,2]$~\cite{FF21, KW21} (see~\cite[Lemma 3]{SW25_channels} for an explicit proof). This implies that the $k$-unextendible geometric R\'enyi divergence of channels is subadditive under tensor products (see Appendix~\ref{app:n_shot_priv_cap_ub_proof} for a complete proof). That is,
\begin{equation}
	\widehat{E}^{\alpha}_k\!\left(\mathcal{N}^{\otimes n}_{A\to B}\right) \le n\widehat{E}^{\alpha}_k\!\left(\mathcal{N}_{A\to B}\right) \qquad \forall n\in \mathbb{N}.
\end{equation}
This equips us with tools to obtain a single-letter upper bound on the $n$-shot, forward-assisted private capacity of a channel, which we state in Corollary~\ref{cor:n_shot_priv_cap_ub} below.

\begin{corollary}\label{cor:n_shot_priv_cap_ub}
	Fix $\alpha \in (1,2]$ and integer $k\ge 2$. For a given channel $\mathcal{N}_{A\to B}$, $\varepsilon \in [0,1]$, and  $n\in \mathbb{N}$, if
	\begin{equation}
		n\widehat{E}^{\alpha}_k\!\left(\mathcal{N}_{A\to B}\right) + \frac{\alpha}{\alpha - 1}\log_2\!\left(\frac{1}{1-\varepsilon}\right) \le \log_2k,
	\end{equation}
	then
	\begin{equation}
		P^{\varepsilon,\to}\!\left(\mathcal{N}^{\otimes n}_{A\to B}\right) \le \log_2\!\left(\frac{k-1}{k}\right) - \log_2\!\left(2^{-n\widehat{E}^{\alpha}_k\left(\mathcal{N}\right)}(1-\varepsilon)^{\frac{\alpha}{\alpha - 1}} - \frac{1}{k}\right).
	\end{equation}
	
\end{corollary}
\begin{proof}
	See Appendix~\ref{app:n_shot_priv_cap_ub_proof}.
\end{proof}

The $k$-unextendible geometric R\'enyi divergence of channels can be computed for rational values of $\alpha \in (1,2]$ by means of a semidefinite program. As such, the upper bound on the $n$-shot, forward-assisted private capacity of a channel given in Corollary~\ref{cor:n_shot_priv_cap_ub} is efficiently computable. In Appendix~\ref{app:SDPs}, we present the semidefinite program to compute the $k$-unextendible geometric R\'enyi divergence of a channel for $\alpha = 1+2^{-\ell}$ and $\ell \in \mathbb{N}$. We refer the reader to~\cite{FS17} for a detailed discussion on this topic.

The $\alpha$-sandwiched R\'enyi relative entropy of channels is not generally additive or subadditive under tensor products (see \cite{Fang2020} for counterexamples in the limit $\alpha \to 1$), which prevents us from replacing the $k$-unextendible geometric R\'enyi divergence of channels with the $k$-unextendible sandwiched R\'enyi divergence of channels in Corollary~\ref{cor:n_shot_priv_cap_ub}. However, it was shown in~\cite[Theorem 6]{TWW17} that the Rains information of a channel, induced by the $\alpha$-sandwiched R\'enyi relative entropy, is weakly subadditive. 

Following techniques from~\cite{TWW17}, we find that the $k$-unextendible sandwiched R\'enyi divergence of channels also obeys weak subadditivity, as stated in Proposition~\ref{prop:weak_subadd_sandwich} below.
\begin{proposition}\label{prop:weak_subadd_sandwich}
    Let $\mathcal{N}_{A\to B}$ be a quantum channel. Fix $\alpha \in (1,\infty)$, $k\ge 2$, and $n\in \mathbb{N}$. Let $\rho_{RA^n}$ be an arbitrary state. Then, the following inequality holds:
    \begin{equation}
        \widetilde{E}^{\alpha}_k\!\left(\mathcal{N}^{\otimes n}_{A\to B}(\rho_{RA^n})\right) \le n\widetilde{E}^{\alpha}_k\!\left(\mathcal{N}_{A\to B}\right) + \frac{\alpha}{\alpha - 1}\log_2\!\left(\binom{n+|A|^2-1}{n}\right).
    \end{equation}
\end{proposition}
\begin{proof}
    See Appendix~\ref{app:weak_subadd_sandwich_proof}.
\end{proof}

\medskip

We can now use Proposition~\ref{prop:weak_subadd_sandwich} to obtain a single-letter upper bound on the $n$-shot, forward assisted private capacity of a channel in terms of its $k$-unextendible sandwiched R\'enyi divergence. We state this bound  in Corollary~\ref{cor:n_shot_priv_cap_sandwich} below:
\begin{corollary}\label{cor:n_shot_priv_cap_sandwich}
    Fix $\alpha \in (1,\infty)$ and integer $k\ge 2$. For a given channel $\mathcal{N}_{A\to B}$, $\varepsilon\in [0,1]$, and  $n\in \mathbb{N}$, if
	\begin{equation}
		n\widetilde{E}^{\alpha}_k\!\left(\mathcal{N}_{A\to B}\right) + \frac{\alpha}{\alpha - 1}\log_2\!\left(\frac{C(n,|A|)}{1-\varepsilon}\right) \le \log_2k,
	\end{equation}
	where
    \begin{equation}
        C(n,|A|) \coloneqq \binom{n+|A|^2-1}{n} ,
    \end{equation}
    then
	\begin{equation}
		P^{\varepsilon,\to}\!\left(\mathcal{N}^{\otimes n}_{A\to B}\right) \le \log_2\!\left(\frac{k-1}{k}\right) - \log_2\!\left(2^{-n\widetilde{E}^{\alpha}_k\left(\mathcal{N}\right)}\!\left(\frac{1-\varepsilon}{C(n,|A|)}\right)^{\frac{\alpha}{\alpha - 1}} - \frac{1}{k}\right).
	\end{equation}
\end{corollary}
\begin{proof}
    See Appendix~\ref{app:n_shot_priv_cap_sandwich_proof}.
\end{proof}

\begin{remark}
    The $n$-shot, forward-assisted quantum capacity of a channel can never be larger than the $n$-shot, forward-assisted private capacity of the channel. This is because every quantum communication protocol can be transformed into a private communication protocol by simply transmitting one share of a maximally entangled state and both parties measuring their respective systems in the computational basis. Therefore, the quantities in Theorem~\ref{theo:priv_cap_ub_hypo_test}, Corollary~\ref{cor:n_shot_priv_cap_ub}, and Corollary~\ref{cor:n_shot_priv_cap_sandwich} also serve as upper bounds on the one-shot, forward-assisted quantum capacity of a channel.
\end{remark}

\section{Numerical examples}\label{sec:applications}

In this section, we demonstrate the results obtained in Sections~\ref{sec:one_way_key_distill} and~\ref{sec:one_way_priv_cap} with some numerical examples. 

First, we compute several quantities of interest in the context of one-way secret-key distillation using isotropic states. In particular, we compute upper bounds on the one-shot, one-way distillable key and the $n$-shot, one-way distillable key rate of isotropic states in Section~\ref{sec:n_shot_iso_key}. In Section~\ref{sec:iso_single_bit}, we compute a lower bound on the minimum number of secret bits needed to distill a single secret with some fixed error tolerance. 

Next, we numerically demonstrate the results from Section~\ref{sec:one_way_priv_cap} for erasure channels. We compute upper bounds on the $n$-shot, forward-assisted private capacity in Section~\ref{sec:n_shot_eras_cap}, and we compute the minimum number of uses of an erasure channel needed to transmit a single secret bit using a one-way LOCC protocol in Section~\ref{sec:eras_ch_min_use}.

\begin{figure}
    \centering
    \includegraphics[width=0.65\linewidth]{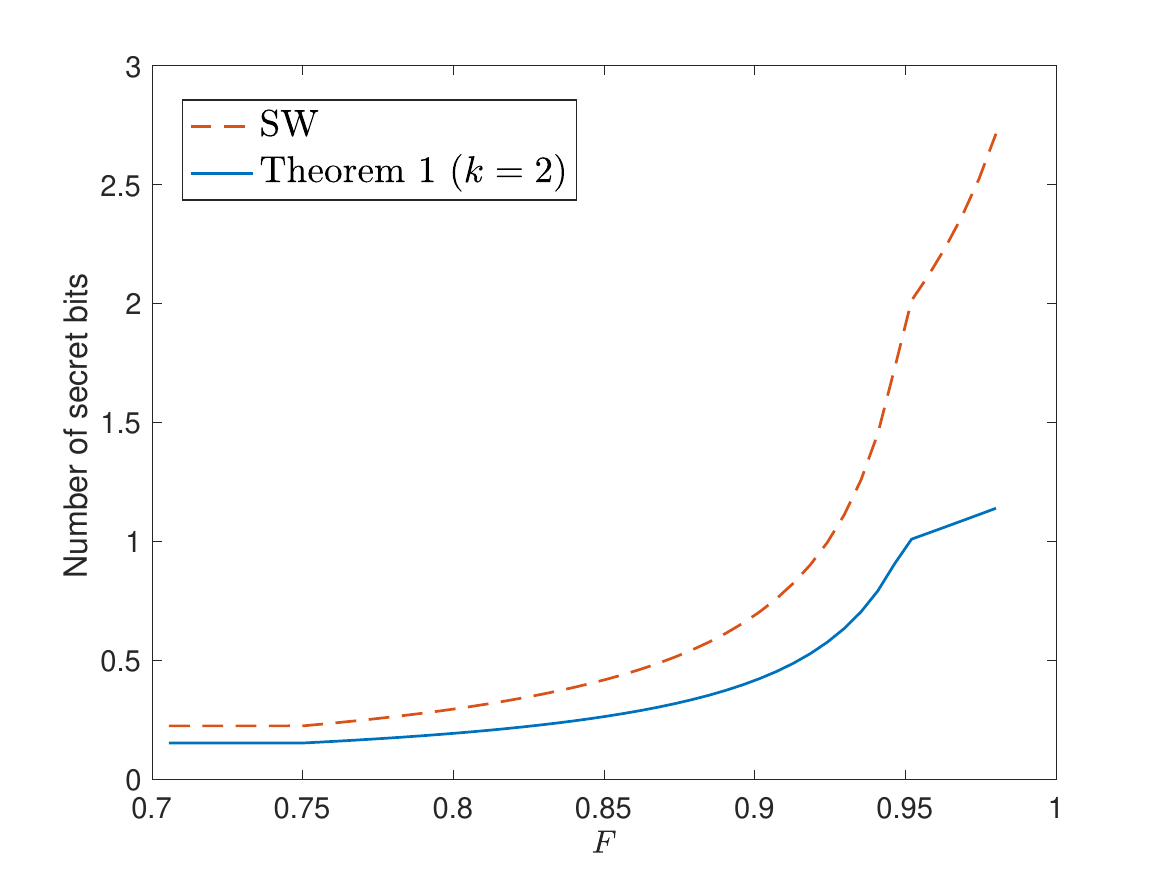}
    \caption{Upper bound on the number of secret bits that can be distilled from a single copy of an isotropic state with $\varepsilon = 0.05$. The bound from Theorem~\ref{theo:dd_1shot_key_bnd_k_ext} is compared against the bound from~\cite[Theorem 2]{SW25} for different values of the parameter $F$ of the isotropic state (see~\eqref{eq:iso_st_defn} for reference).}
    \label{fig:iso_unext_vs_unjoin}
\end{figure}

\subsection{\texorpdfstring{$n$}{n}-Shot, one-way distillable key of isotropic states}\label{sec:n_shot_iso_key}

We first demonstrate our upper bounds on the one-shot, one-way distillable key of an isotropic state using Theorem~\ref{theo:dd_1shot_key_bnd_k_ext}. Recall that a $d$-dimensional isotropic state is defined for a parameter $F\in [0,1]$ as follows~\cite{HH99}:
\begin{equation}\label{eq:iso_st_defn}
	\zeta^{F,d}_{AB} \coloneqq F\Phi^d_{AB} + (1-F)\frac{I_{AB}-\Phi^d_{AB}}{d^2-1}.
\end{equation}
In Figure~\ref{fig:iso_unext_vs_unjoin}, we plot the numerical values of the upper bound on the one-shot, one-way distillable key of an isotropic state obtained for different values of the parameter $F$, with $\varepsilon = 0.05$ and $k=2$. We compare our bounds with the analogous bound obtained in~\cite[Theorem 2]{SW25}, and we observe that the bound from Theorem~\ref{theo:dd_1shot_key_bnd_k_ext} performs better than the bound from~\cite[Theorem 2]{SW25} for this example.

Next, we turn our attention to the task of distilling a secret key from $n$ i.i.d.~copies of an isotropic state. While we can use Corollary~\ref{cor:DD_key_n_shot_ub_sandwich} to compute an upper bound on the $n$-shot, one-way distillable key of an isotropic state, here we exploit the symmetries of isotropic states to compute this quantity using Theorem~\ref{theo:dd_1shot_key_bnd_k_ext} itself in a computationally feasible way.

In~\cite[Proposition 1]{KDWW21}, the authors showed that the $k$-unextendible generalized divergence of an isotropic state $\zeta^{F,d}_{AB}$ is equal to the following:
\begin{equation}
	\mathbf{E}_k\!\left(\zeta^{F,d}_{AB}\right) = \inf_{G\in \left[0,\frac{1}{d} + \frac{1}{k} - \frac{1}{dk}\right]}\mathbf{D}\!\left(\kappa(F)\middle\Vert\kappa(G)\right),
\end{equation}
where 
\begin{equation}
	\kappa(F) \coloneqq F|0\rangle\!\langle 0| + (1-F)|1\rangle\!\langle 1|.
\end{equation}
Furthermore, in~\cite[Section V.A]{KDWW21}, the authors argued that the $k$-unextendible hypothesis-testing divergence of $n$ i.i.d.~copies of an isotropic state is bounded from above by the following:
\begin{equation}\label{eq:n_iso_hypo_test_ub_bernoulli}
	E^{\varepsilon}_k\!\left(\left(\zeta^{F,d}_{AB}\right)^{\otimes n}\right) \le \inf_{G\in \left[0,\frac{1}{d} + \frac{1}{k} - \frac{1}{dk}\right]} D^{\varepsilon}_H\!\left(\left\{F,1-F\right\}^{\times n}
    \middle \Vert\left\{G,1-G\right\}^{\times n}\right),
\end{equation}
where the quantity on the right hand side of the above equality is the hypothesis-testing relative entropy between two Bernoulli distributions. The hypothesis-testing relative entropy between two Bernoulli distributions can be computed using a linear program, which greatly reduces the cost of computing the upper bound on the $n$-shot, one-way distillable key of an isotropic state using Theorem~\ref{theo:dd_1shot_key_bnd_k_ext}. Fixing $G = \frac{1}{d} + \frac{1}{k} - \frac{1}{dk}$, we can eliminate any dependence of the choice of $k$ on the computational complexity of computing an upper bound on the $n$-shot, one-way distillable key of an isotropic state using Theorem~\ref{theo:dd_1shot_key_bnd_k_ext} and~\eqref{eq:n_iso_hypo_test_ub_bernoulli}.

Recall that setting $k\to \infty$ in Theorem~\ref{theo:dd_1shot_key_bnd_k_ext} leads to the hypothesis-testing relative entropy of entanglement, which is a well-known bound on the one-shot distillable key of a state, and hence, an upper bound on the one-shot, one-way distillable key of the state as well. From~\eqref{eq:n_iso_hypo_test_ub_bernoulli}, one can verify that
\begin{equation}\label{eq:hypo_test_rel_entanglement_iso_n_copy}
	E^{\varepsilon}_{k = \infty}\!\left(\left(\zeta^{F,d}_{AB}\right)^{\otimes n}\right) \le  D^{\varepsilon}_H\!\left(\left\{F,1-F\right\}^{\times n}\middle \Vert \left\{\frac{1}{d},1-\frac{1}{d}\right\}^{\times n}\right).
\end{equation}
This relaxation allows us to numerically compare our bounds with the hypothesis-testing relative entropy of entanglement bound on the $n$-shot, one-way distillable key of a state, obtained in~\cite{WTB17}.

\begin{figure}
	\centering
	\includegraphics[width = 0.7\linewidth]{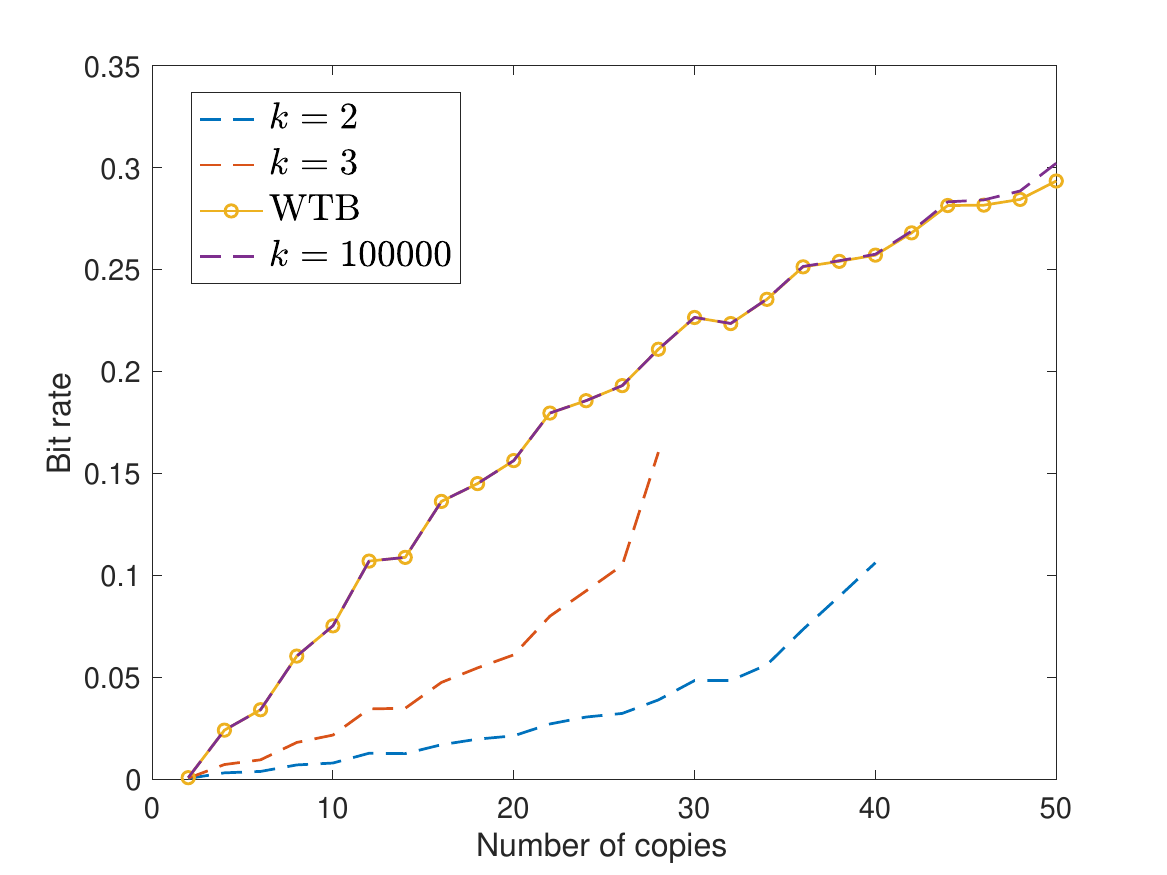}
	\caption{Upper bounds on the $n$-shot, one-way distillable key rate of a two-dimensional isotropic state with $F = 0.95$ and $\varepsilon = 10^{-5}$. The bounds are computed for different values of $k$ using Theorem~\ref{theo:dd_1shot_key_bnd_k_ext}, and they are compared against the hypothesis-testing relative entropy of entanglement bound. The bounds from Theorem~\ref{theo:dd_1shot_key_bnd_k_ext} can only be computed for a finite number of copies of the state, say $n$, since $E^{\varepsilon}_k\!\left(\rho^{\otimes n}\right)$ must be less than $\log_2 k$ for the bound to hold. This restriction manifests itself in the plot as the curves corresponding to $k=2$ and $k=3$ end abruptly.}
	\label{fig:iso_dist_key_hypo_rate}
\end{figure}

In Figure~\ref{fig:iso_dist_key_hypo_rate}, we plot the upper bounds on the rate of distilling secret bits from $n$ copies of an isotropic state using Theorem~\ref{theo:dd_1shot_key_bnd_k_ext}, along with~\eqref{eq:n_iso_hypo_test_ub_bernoulli}, for $k=2, 3$, and $10^5$. We also plot an upper bound on the hypothesis-testing relative entropy of entanglement of $n$ copies of the isotropic state using~\eqref{eq:hypo_test_rel_entanglement_iso_n_copy}, which is a well-known upper bound on the $n$-shot distillable key of a state~\cite{WTB17} and is also achieved by setting $k\to\infty$ in Theorem~\ref{theo:dd_1shot_key_bnd_k_ext}. In this example, we find that fixing $k=2$ in Theorem~\ref{theo:dd_1shot_key_bnd_k_ext} gives the tightest bound on the $n$-shot, one-way distillable key of the isotropic state.

Recall that the bound in Theorem~\ref{theo:dd_1shot_key_bnd_k_ext} holds for a fixed $k$ if and only if the $k$-unextendible hypothesis testing divergence of the state is less than $\log_2 k$. This is reflected in Figure~\ref{fig:iso_dist_key_hypo_rate} as the upper bound from Theorem~\ref{theo:dd_1shot_key_bnd_k_ext} can only be computed for a small number of copies when $k$ is set equal to two or three. One can try to compute the bound from Theorem~\ref{theo:dd_1shot_key_bnd_k_ext} for larger number of copies by choosing a large $k$, as we show in Figure~\ref{fig:iso_dist_key_hypo_rate} by setting $k=10^5$. However, such a choice may significantly worsen the bound. In fact, for the example of isotropic state considered in Figure~\ref{fig:iso_dist_key_hypo_rate}, the bound with $k=10^5$ is less tight than the hypothesis-testing relative entropy of entanglement bound ($k\to\infty$) for fifty copies of an isotropic state.

\subsubsection{Minimum number of copies to distill a single secret bit}\label{sec:iso_single_bit}

The numerical example discussed in Figure~\ref{fig:iso_dist_key_hypo_rate} shows that one needs several copies of an isotropic state before a single secret bit can be distilled using any one-way LOCC protocol. This feature can be seen in the state of the art key distillation protocols, where several copies of an isotropic state are needed before a single secret bit can be distilled despite the protocols achieving significantly higher key rates asymptotically~\cite{TL17}. This motivates the question: What is the minimum number of copies of an isotropic state needed to distill a single secret bit using a  one-way LOCC protocol?

\begin{figure}
	\centering
	\includegraphics[width = 0.65\linewidth]{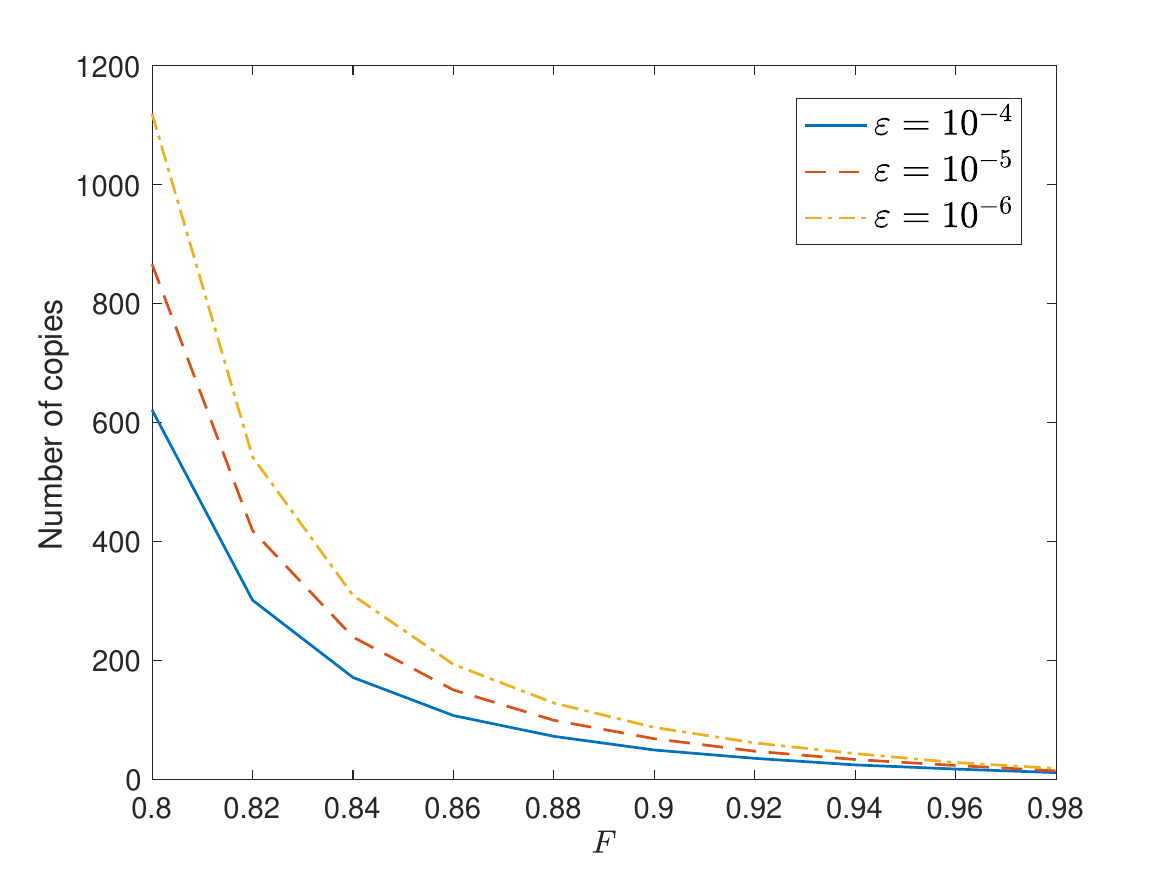}
	\caption{Lower bound on the minimum number of copies of a two-dimensional isotropic state needed to distill a single secret bit, with error tolerance $\varepsilon$, using a one-way LOCC protocol. The lower bound on the minimum number of copies is computed for different values of the parameter $F$ (see~\eqref{eq:iso_st_defn}) and three different values of the error tolerance $\varepsilon$. When $F=1$, the isotropic state is a maximally entangled state, and only a single copy of the state would suffice to distill a secret bit with any $\varepsilon \in [0,1]$. However, the isotropic state becomes increasingly noisy as $F$ decreases, which means that a larger number of copies are needed to distill a single secret bit with the desired error tolerance. }
	\label{fig:min_copies_iso_secret_key}
\end{figure}

The statement of Theorem~\ref{theo:dd_1shot_key_bnd_k_ext}, along with the inequality in~\eqref{eq:n_iso_hypo_test_ub_bernoulli}, allows us to obtain a lower bound on the minimum number of copies of an isotropic state required to distill a single secret bit with some error tolerance $\varepsilon$ using a one-way LOCC protocol. Essentially, we use a binary search to find the smallest $n$ such that the upper bound on $K^{\varepsilon,\to}\!\left(\left(\zeta^{F,d}_{AB}\right)^{\otimes n}\right)$ obtained from Theorem~\ref{theo:dd_1shot_key_bnd_k_ext} is greater than or equal to one. We demonstrate our lower bounds on the minimum number of copies of an isotropic state required to distill a single secret bit using a one-way LOCC protocol in Figure~\ref{fig:min_copies_iso_secret_key}.

We note that the single-letter upper bounds on the $n$-shot, one-way distillable key of a state obtained in Corollary~\ref{cor:DD_key_n_shot_ub_sandwich} yield a simpler bound on the minimum number of copies of a given quantum state needed to distill a single secret bit using a one-way LOCC protocol. For isotropic states, these bounds are much worse than the bounds demonstrated in Figure~\ref{fig:min_copies_iso_secret_key} owing to the fact that Corollary~\ref{cor:DD_key_n_shot_ub_sandwich} are relaxations of the statement in Theorem~\ref{theo:dd_1shot_key_bnd_k_ext}. However, these bounds have the advantage of being efficiently computable for all states, not just highly symmetric states like the isotropic states.

\subsection{Private communication over erasure channels}

A commonly studied class of channels in the context of quantum and private communication are erasure channels. The action of an erasure channel $\mathcal{E}^p_{A\to B}$ on an arbitrary state $\rho_{RA}$ is mathematically described as follows~\cite{GBP97}:
\begin{equation}\label{eq:eras_ch_defn}
    \mathcal{E}^p_{A\to B}\!\left(\rho_{RA}\right) = (1-p)\rho_{RA} + p\operatorname{Tr}_{A}\!\left[\rho_{RA}\right]\otimes |e\rangle\!\langle e|_B,
\end{equation}
where $p$ refers to the erasure probability and $|e\rangle_B$ is the erasure symbol, which is orthogonal to every vector in the Hilbert space $\mathcal{H}_B$. 

\subsubsection{\texorpdfstring{$n$}{n}-Shot, forward-assisted  private capacity of erasure channels}\label{sec:n_shot_eras_cap}

Now, we demonstrate our upper bounds on the $n$-shot private capacity of an erasure channel using Theorem~\ref{theo:priv_cap_ub_hypo_test}.  

Note that a $d$-dimensional erasure channel with an erasure probability greater than or equal to $1-\frac{1}{k}$ is $k$-extendible. This can be verified from its Choi state
\begin{equation}
    \Phi^{\mathcal{E}^{1-1/k}}_{AB} = \frac{1}{k}\Phi^d_{AB} + \left(1-\frac{1}{k}\right)\frac{I_A}{|A|}\otimes |e\rangle\!\langle e|_B,
\end{equation}
which has the following $k$-extension:
\begin{equation}
    \sigma^{\mathcal{E}^{1-1/k}}_{AB_{[k]}} \coloneqq \frac{1}{k}\sum_{i=1}^k\Phi^d_{AB_i}\otimes \bigotimes_{\substack{j=1,\\ i\neq j}}^k|e\rangle\!\langle e|_{B_j}.
\end{equation}
Therefore, 
\begin{equation}
    \mathbf{E}_k\!\left(\mathcal{E}^{p}_{A\to B}\right) \le \mathbf{D}\!\left(\mathcal{E}^p_{A\to B}\middle\Vert\mathcal{E}^{1-1/k}_{A\to B}\right),
\end{equation}
which follows from the definition of $k$-unextendible generalized divergence of channels. 

\begin{proposition}\label{prop:eras_ch_gen_div}
    The $k$-unextendible generalized divergence of a tensor product of erasure channels is bounded from above by the following:
    \begin{equation}
        \mathbf{E}_k\!\left(\left(\mathcal{E}^p_{A\to B}\right)^{\otimes n}\right) \le \mathbf{D}\!\left(\left\{1-p,p\right\}^{\times n}\middle\Vert\left\{\frac{1}{k},1-\frac{1}{k}\right\}^{\times n}\right).
    \end{equation}
\end{proposition}

\begin{proof}
    See Appendix~\ref{app:eras_ch_gen_div}.
\end{proof}

\medskip

\begin{figure}
    \centering
    \begin{subfigure}{0.45\linewidth}
        \includegraphics[width=\linewidth]{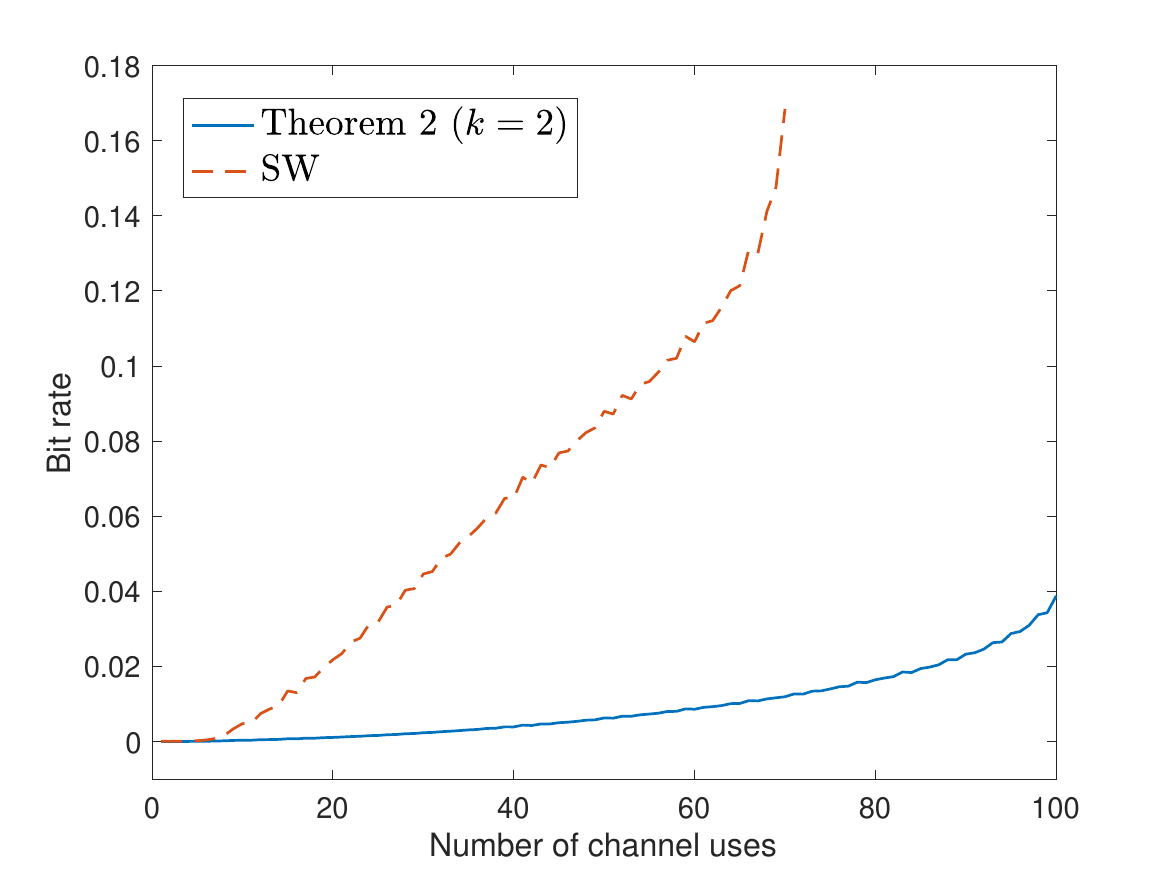}
        \caption{\centering Rate of private communication}
        \label{fig:eras_ch_key_rate}
    \end{subfigure}
    \begin{subfigure}{0.45\linewidth}
        \includegraphics[width=\linewidth]{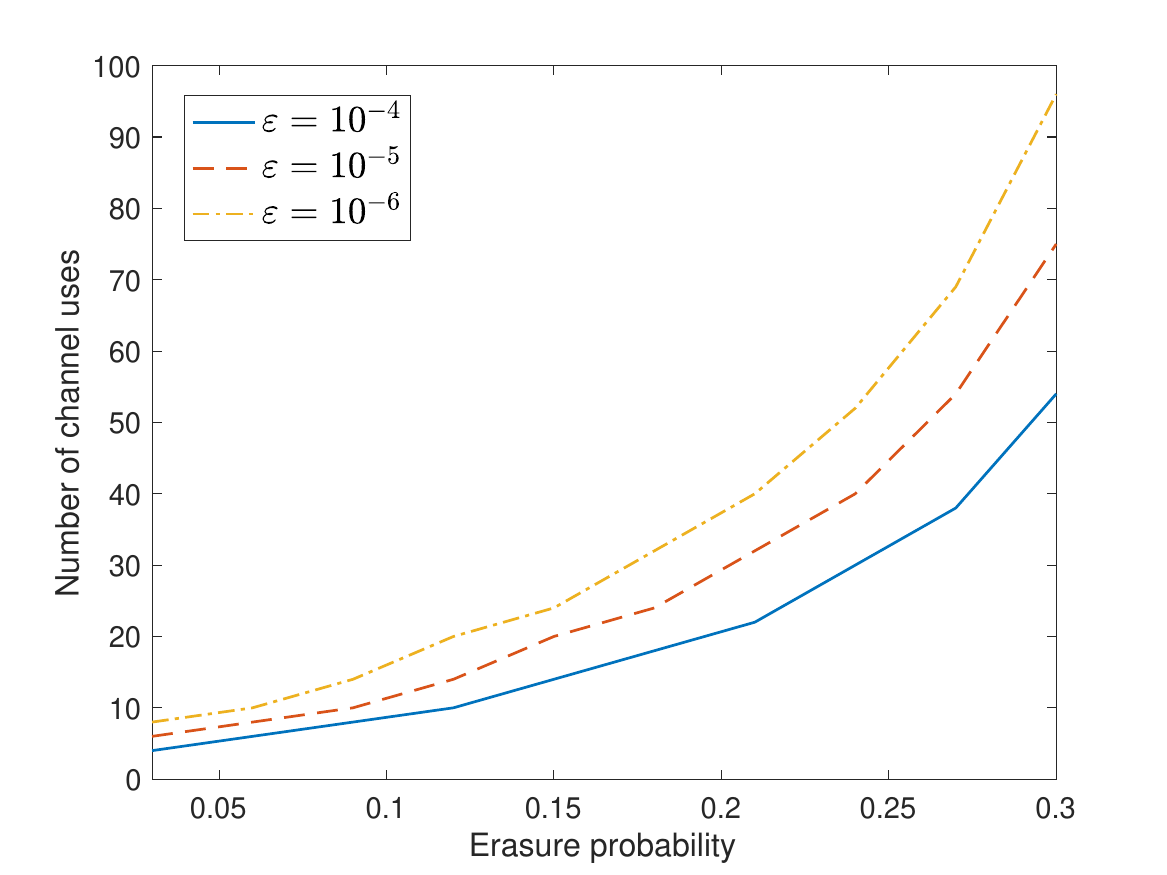}
        \caption{\centering Single bit transmission cost}
        \label{fig:eras_ch_min_uses}
    \end{subfigure}
    \caption{\textbf{(a)} Upper bound on the $n$-shot, forward-assisted private capacity of an erasure channel using Theorem~\ref{theo:priv_cap_ub_hypo_test} and~\cite[Theorem 4]{SW25}. The erasure probability is set equal to 0.3 and the error tolerance is set equal to $10^{-5}$. The upper bound from~\cite[Theorem 4]{SW25} holds for only 70 channel uses for this choice of parameters. However, the bound from Theorem~\ref{theo:priv_cap_ub_hypo_test} holds for 104 channel uses with this choice of parameters.
    \textbf{(b)} Lower bound on the minimum number of uses of an erasure channel needed to securely transmit a single bit over the channels, assisted by local operations and forward public communication.}
    \label{fig:eras_ch_bounds}
    
\end{figure}

Proposition~\ref{prop:eras_ch_gen_div} allows us to compute an upper bound on the $n$-shot, forward-assisted private capacity of an erasure channel using Theorem~\ref{theo:priv_cap_ub_hypo_test}, efficiently in $n$.

In Figure~\ref{fig:eras_ch_key_rate}, we plot upper bounds on the rate at which bits can be transmitted securely over multiple uses of an erasure channel. We chose the erasure probability to be equal to 0.3 and the error tolerance was set equal to $10^{-5}$. We compare our bounds against the upper bound on the one-shot, forward-assisted private capacity discovered in~\cite[Theorem 4]{SW25}. To compute the smooth-min unextendible entanglement of a tensor product of erasure channels, which is the quantity of interest in~\cite[Theorem 4]{SW25}, we use the fact that $\left(\mathcal{E}^p_{A\to B}\right)^{\otimes n}$ can be joined with $\left(\mathcal{E}^{1-p}_{A\to B}\right)^{\otimes n}$. This may not be an optimal choice, but it yields an upper bound on the quantity in~\cite[Theorem 4]{SW25} for $\left(\mathcal{E}^p_{A\to B}\right)^{\otimes n}$, facilitating a comparison with our bounds.

We did not compare our bounds with~\cite{WTB17} for this example because the bound from~\cite{WTB17} is much higher than the bound from Theorem~\ref{theo:priv_cap_ub_hypo_test} for this example. 

\subsubsection{Minimum number of channel uses to securely transmit a single bit}\label{sec:eras_ch_min_use}

In this section, we use Theorem~\ref{theo:priv_cap_ub_hypo_test} to compute a lower bound on the minimum number of uses of an erasure channel needed to transmit a single bit securely over the channels, assisted by local operations and an arbitrary amount of forward public communication.

Similar to our approach in Section~\ref{sec:iso_single_bit}, we employ binary search to find the smallest positive integer $n$ such that the upper bound on $P^{\varepsilon,\to}\!\left(\left(\mathcal{E}^p_{A\to B}\right)^{\otimes n}\right)$, obtained using Theorem~\ref{theo:priv_cap_ub_hypo_test} along with Proposition~\ref{prop:eras_ch_gen_div}, is greater than or equal to one. In Figure~\ref{fig:eras_ch_min_uses}, we plot our lower bounds on the minimum number of uses of an erasure channel to transmit a single bit securely over the channels, with the assistance of local operations and forward public communication, for different values of erasure probability $p$ and error tolerance $\varepsilon$.

\section{Conclusion}

\subsection{Summary}

In this paper, we determined the maximum probability with which a $k$-extendible state can pass a privacy test, and we found it to be equal to the maximum fidelity between a $k$-extendible state and the standard maximally entangled state.

As an application of our finding, we showed that the limits on quantum communication under freely available local operations and forward classical communication found in~\cite{KDWW19,KDWW21} are in fact limits on private communication under the same set of free operations, generalizing the results of~\cite{KDWW19,KDWW21}. As a consequence, we obtained upper bounds on the one-shot, one-way distillable key of a state and the one-shot, forward-assisted private capacity of a channel, which can be computed using a semidefinite program.

We also extended our formalism to the $n$-shot setting and obtained single-letter upper bounds on the $n$-shot, one-way distillable key of a state and the $n$-shot, forward-assisted private capacity of a channel, both of which can be computed using a semidefinite program. 

\subsection{Future directions}

A question that still remains unanswered is as follows: For a collection of joinable states, what are the maximum probabilities with which each of them can pass a privacy test? Since the maximum fidelity of a $k$-extendible state with the standard maximally entangled state turned out to be the maximum probability with which a $k$-extendible state passes the privacy test, one might expect a similar result to hold for joinable states as well.  

We were able to numerically demonstrate that the bounds on the one-shot, one-way distillable key of a state and the one-shot, forward-assisted private capacity of a state based on extendibility, obtained in this work, performed better than the bounds obtained in~\cite{SW25}, which were based on unjoinability. However, a stronger understanding of the connections between the resource theory of $k$-unextendibility developed in~\cite{KDWW19,KDWW21} and the resource theory of unextendible entanglement~\cite{WWW24} may shed light onto the regimes in which either bound performs better than the other.

\section*{Acknowledgements}

VS  thanks Ian George, Theshani Nuradha, Ernest Y.-Z.~Tan, and Marco Tomamichel for helpful discussions. The authors also thank the organizers of Quantum Resources 2025, held in Jeju, Korea, for organizing the conference and especially the open problems session, which served as a catalyst for this paper.

VS thanks the Dieter Schwarz Exchange Programme on Quantum Communication and Security at the Centre for Quantum Technologies for support. 
KH acknowledges support of the National Science Centre, Poland, under grant Opus 25, 2023/49/B/ST2/02468.
AP acknowledges support from the National Science Centre Poland (Grant No. 2022/46/E/ST2/00115). MMW acknowledges support
from the National Science Foundation under grant no.~2329662 and
from the Cornell School of Electrical and Computer Engineering.

\appendix

\appendix

\section{Semidefinite programs}\label{app:SDPs}

In this section we present all the semidefinite programs that were used in this work.

\begin{enumerate}
    \item \textbf{$k$-unextendible hypothesis-testing divergence of a state:}
    \begin{equation}
        E^{\varepsilon}_{k}(\rho_{AB}) = -\log_2\max\left\{
        \begin{array}{c}
             \mu(1-\varepsilon)-\operatorname{Tr}\!\left[Z_{AB}\right]: \\
             \mu \ge 0, Z_{AB} \ge 0, \sigma_{AB_{[k]}} \ge 0,\\
             \mu\rho_{AB} \le \operatorname{Tr}_{B_{[k]\setminus 1}}\!\left[\sigma_{AB_{[k]}}\right] + Z_{AB},\\
             W^{\pi}_{B_{[k]}}\sigma_{AB_{[k]}}\!\left(W^{\pi}_{B_{[k]}}\right)^{\dagger} = \sigma_{AB_{[k]}}\quad \forall \pi \in S_k,\\
             \operatorname{Tr}\!\left[\sigma_{AB_{[k]}}\right] = 1
        \end{array}
        \right\}.
    \end{equation}
    \item \textbf{$k$-unextendible max divergence of a state:} \begin{equation}\label{eq:max_unext_ent_st_SDP}
        E^{\max}_k\!\left(\rho_{AB}\right) = -\log_2 \max\left\{
        \begin{array}{c}
             \lambda:  \\
             \lambda \rho_{AB} \le \operatorname{Tr}_{B}\!\left[\sigma_{AB_{[k]}}\right],\\
             \sigma_{AB_{[k]}} \ge 0,\\
             W^{\pi}_{B_{[k]}}\sigma_{AB_{[k]}}\!\left(W^{\pi}_{B_{[k]}}\right)^{\dagger} = \sigma_{AB_{[k]}}\quad \forall \pi \in S_k,\\
             \operatorname{Tr}\!\left[\sigma_{AB_{[k]}}\right] = 1
        \end{array}
        \right\}.
    \end{equation}
    \item \textbf{$k$-unextendible hypothesis-testing divergence of a channel:} The hypothesis-testing relative entropy of a channel $\mathcal{N}$ with respect to a channel $\mathcal{M}$ has a semidefinite program, which was given in~\cite[Proposition~2]{WW19}. We use it to write the semidefinite program for the $k$-unextendible hypothesis-testing divergence of a channel as follows:
    \begin{equation}
        E^{\varepsilon}_{k}\!\left(\mathcal{N}_{A\to B}\right) = -\log_2 \max\left\{
        \begin{array}{c}
            \mu(1-\varepsilon) - \lambda: \\
             \lambda \ge 0, \mu \ge 0, Y_{AB}\ge 0,\Gamma^{\mathcal{P}}_{AB_{[k]}}\ge 0,\\
             \mu\Gamma^{\mathcal{N}}_{AB} \le \operatorname{Tr}_{B_{[k]\setminus 1}}\!\left[\Gamma^{\mathcal{P}}_{AB_{[k]}}\right]+Y_{AB}\\
             \operatorname{Tr}_{B}\!\left[Y_{AB}\right] \le \lambda I_A,\\
             W^{\pi}_{B_{[k]}}\Gamma^{\mathcal{P}}_{AB_{[k]}}\!\left(W^{\pi}_{B_{[k]}}\right)^{\dagger} = \Gamma^{\mathcal{P}}_{AB_{[k]}}\quad \forall \pi \in S_k,\\
             \operatorname{Tr}_{B_{[k]}}\!\left[\Gamma^{\mathcal{P}}_{AB_{[k]}}\right] = I_A
        \end{array}
        \right\},
    \end{equation}
    where $\Gamma^{\mathcal{N}}_{AB}$ is the Choi operator of the channel $\mathcal{N}_{A\to B}$ defined in~\eqref{eq:Choi_op_defn}.
    
        \item \textbf{$k$-unextendible geometric R\'enyi divergence of a channel:} Fix $\ell \in \mathbb{N}$. The $k$-unextendible geometric R\'enyi divergence of a channel $\mathcal{N}_{A\to B}$ for $\alpha = 1+2^{-\ell}$ can be computed using the following semidefinite program:
    \begin{equation}
		\widehat{E}^{\alpha}_k\!\left(\mathcal{N}_{A\to B}\right) = 2^\ell \min_{\substack{y\in \mathbb{R}, \Gamma^{\mathcal{P}}_{AB_{[k]}} \ge 0\\M_{AB}, \left\{N^i_{AB}\right\}_{i=0}^{\ell},\in \operatorname{Herm}}} \log_2 y,
	\end{equation}
	subject to the constraints,
	\begin{align}
		W^{\pi}_{B_{[k]}}\Gamma^{\mathcal{P}}_{AB_{[k]}}\!\left(W^{\pi}_{B_{[k]}}\right)^{\dagger} &= \Gamma^{\mathcal{P}}_{AB_{[k]}}\quad \forall \pi \in S_k,\\
             \operatorname{Tr}_{B_{[k]}}\!\left[\Gamma^{\mathcal{P}}_{AB_{[k]}}\right] &= I_A,\\
		\operatorname{Tr}_{B} \left[M_{AB}\right] &\le yI_A\label{eq:geo_unext_ent_cond_1},\\
		\operatorname{Tr}_{B_{[k]\setminus 1}} \left[\Gamma^{\mathcal{P}}_{AB_{[k]}}\right] &= N^0_{AB}, \label{eq:geo_unext_ent_cond_2}\\
		\begin{bmatrix}
			M_{AB} & \Gamma^{\mathcal{N}}_{AB}\\
			\Gamma^{\mathcal{N}}_{AB} & N^{\ell}_{AB}
		\end{bmatrix}
		&\ge 0,\label{eq:geo_unext_ent_cond_3}\\
		\begin{bmatrix}
			\Gamma^{\mathcal{N}}_{AB} & N^i_{AB}\\
			N^i_{AB} & N^{i-1}_{AB}
		\end{bmatrix}
		&\ge 0 \quad \forall i\in \{1,2,\ldots,\ell\},\label{eq:geo_unext_ent_cond_4}
	\end{align}
 where $\Gamma^{\mathcal{N}}_{AB}$ is the Choi operator of the channel $\mathcal{N}_{A\to B}$. To compute the $\alpha$-geometric unextendible entanglement of the channel for other rational values of $\alpha$ see~\cite[Table 4]{FS17}.
\end{enumerate}

\section{Proof of Theorem~\ref{theo:priv_cap_ub_hypo_test}}\label{app:priv_cap_ub}

In this section, we present the proof of Theorem~\ref{theo:priv_cap_ub_hypo_test}.

Recall the inequality in~\eqref{eq:priv_cap_le_key_st}. Now using Theorem~\ref{theo:dd_1shot_key_bnd_k_ext}, we have
\begin{align}
    P^{\varepsilon,\to}\!\left(\mathcal{N}_{A\to B}\right) &\le \sup_{\rho_{RA}\in \mathcal{S}(RA)} K^{\varepsilon,\to}\!\left(\mathcal{N}_{A\to B}\!\left(\rho_{RA}\right)\right),\\
    &\le \sup_{\rho_{RA}\in \mathcal{S}(RA)}\left\{-\log_2\!\left(2^{-E^{\varepsilon}_k\left(\mathcal{N}\left(\rho_{RA}\right)\right)}- \frac{1}{k}\right) + \log_2\!\left(\frac{k-1}{k}\right)\right\}\\
    &= \log_2\!\left(\frac{k-1}{k}\right) - \inf_{\rho_{RA}\in \mathcal{S}(RA)}\log_2\!\left(2^{-E^{\varepsilon}_k\left(\mathcal{N}\left(\rho_{RA}\right)\right)}- \frac{1}{k}\right)\\
    &= \log_2\!\left(\frac{k-1}{k}\right) - \log_2\!\left(2^{-\sup_{\rho_{RA}\in \mathcal{S}(RA)}E^{\varepsilon}_k\left(\mathcal{N}\left(\rho_{RA}\right)\right)}- \frac{1}{k}\right),\label{eq:priv_cap_le_hypo_test_st}
\end{align}
where the final equality follows from the monotonicity of the logarithm and the exponential functions. 

The statement of Lemma~\ref{lem:unext_div_ch_ge_unext_div_st} implies the following inequality:
\begin{align}
    \sup_{\rho_{RA}\in \mathcal{S}(RA)}E^{\varepsilon}_k\!\left(\mathcal{N}\!\left(\rho_{RA}\right)\right) &\le E^{\varepsilon}_k\!\left(\mathcal{N}_{A\to B}\right)\\
    \implies -\sup_{\rho_{RA}\in \mathcal{S}(RA)}E^{\varepsilon}_k\!\left(\mathcal{N}\!\left(\rho_{RA}\right)\right) &\ge -E^{\varepsilon}_k\!\left(\mathcal{N}_{A\to B}\right) \\
    \implies 2^{-\sup_{\rho_{RA}\in \mathcal{S}(RA)}E^{\varepsilon}_k\!\left(\mathcal{N}\!\left(\rho_{RA}\right)\right)}-\frac{1}{k} &\ge 2^{E^{\varepsilon}_k\!\left(\mathcal{N}_{A\to B}\right)} - \frac{1}{k}\\
    \implies  - \log_2\!\left(2^{-\sup_{\rho_{RA}\in \mathcal{S}(RA)}E^{\varepsilon}_k\left(\mathcal{N}\left(\rho_{RA}\right)\right)}- \frac{1}{k}\right) &\le -\log_2\!\left(2^{-E^{\varepsilon}_k\!\left(\mathcal{N}_{A\to B}\right)} - \frac{1}{k}\right).
\end{align}
Therefore,
\begin{equation}
    P^{\varepsilon,\to}\!\left(\mathcal{N}_{A\to B}\right) \le -\log_2\!\left(2^{-E^{\varepsilon}_k\!\left(\mathcal{N}_{A\to B}\right)} - \frac{1}{k}\right) + \log_2\!\left(\frac{k-1}{k}\right).
\end{equation}
This concludes the proof.

\section{Proof of Corollary~\ref{cor:n_shot_priv_cap_ub}}\label{app:n_shot_priv_cap_ub_proof}

In this section, we present the proof of the Corollary~\ref{cor:n_shot_priv_cap_ub}. 

The $k$-unextendible geometric R\'enyi divergence of a channel is subadditive under tensor product for every $\alpha \in (0,1)\cup (1,2]$. This is easily seen from the following argument: Fix $\alpha \in (0,1)\cup(1,2]$. Let $\mathcal{N}_{A\to B}$ and $\mathcal{M}_{C\to D}$ be arbitrary channels, and let $\mathcal{P}_{A\to B}$ and $\mathcal{Q}_{C\to D}$ be arbitrary $k$-extendible channels. Since a tensor product of $k$-extendible channels is also $k$-extendible, we can write
 \begin{align}
    \widehat{E}^{\alpha}_k\!\left(\mathcal{N}\otimes \mathcal{M}\right) &\le \widehat{D}^{\alpha}\!\left(\mathcal{N}\otimes\mathcal{M}\middle\Vert\mathcal{P}\otimes\mathcal{Q}\right)\\
    &= \widehat{D}^{\alpha}\!\left(\mathcal{N}\middle\Vert\mathcal{P}\right) + \widehat{D}^{\alpha}\!\left(\mathcal{M}\middle\Vert\mathcal{Q}\right),\label{eq:geo_ch_subadditive_1}
\end{align}
where the inequality follows from the definition of $k$-unextendible geometric R\'enyi divergence of channels and the equality follows from the additivity of $\alpha$-geometric R\'enyi relative entropy of channels for every $\alpha\in (0,1)\cup(1,2]$~\cite{FF21,KW21}. Since the inequality in~\eqref{eq:geo_ch_subadditive_1} holds for every $k$-extendible channel $\mathcal{P}_{A\to B}$ and every $k$-extendible channel $\mathcal{Q}_{C\to D}$, we can write
\begin{align}
    \widehat{E}^{\alpha}_k\!\left(\mathcal{N}\otimes \mathcal{M}\right) &\le \inf_{\mathcal{P},\mathcal{Q}\in \operatorname{CPTP}}\left\{\begin{array}{c}
         \widehat{D}^{\alpha}\!\left(\mathcal{N}\middle\Vert\mathcal{P}\right) + \widehat{D}^{\alpha}\!\left(\mathcal{M}\middle\Vert\mathcal{Q}\right):  \\
         \mathcal{P}_{A\to B}\!\left(\Phi_{RA}\right)\in \operatorname{Ext}_k\!\left(R\!:\!B\right), \mathcal{Q}_{C\to D}\!\left(\Phi_{R'C}\right)\in \operatorname{Ext}_k\!\left(R'\!:\!D\right)
    \end{array}\right\}\\
    &= \widehat{E}^{\alpha}_k\!\left(\mathcal{N}_{A\to B}\right) + \widehat{E}^{\alpha}_k\!\left(\mathcal{M}_{C\to D}\right),
\end{align}
which shows that the $k$-unextendible geometric R\'enyi divergence of channels is subadditive under tensor product for every $\alpha \in (0,1)\cup (1,2]$.

The $\alpha$-geometric R\'enyi relative entropy of states is larger than the $\alpha$-sandwiched R\'enyi relative entropy of states for every $\alpha \in \left(\frac{1}{2},1\right)\cup(1,2]$. That is, the following inequality holds for any two states $\rho$ and $\sigma$:
\begin{equation}
    \widehat{D}_{\alpha}\!\left(\rho\middle\Vert\sigma\right) \ge \widetilde{D}_{\alpha}\!\left(\rho\middle\Vert\sigma\right) \qquad \forall \alpha \in \left(\frac{1}{2},1\right)\cup(1,2].
\end{equation}
Consequently,
\begin{equation}
    \widehat{E}^{\alpha}_k\!\left(\rho_{AB}\right) \ge \widetilde{E}^{\alpha}_k\!\left(\rho_{AB}\right)\qquad \forall \alpha \in \left(\frac{1}{2},1\right)\cup(1,2],
\end{equation}
which, when combined with~\eqref{eq:hypo_test_ent_le_sandwich}, leads to the following inequality:
\begin{equation}
    E^{\varepsilon}_k\!\left(\rho_{AB}\right) \le \widehat{E}^{\alpha}_k\!\left(\rho_{AB}\right) + \frac{\alpha}{\alpha - 1}\log_2\!\left(\frac{1}{1-\varepsilon}\right)\qquad \forall \alpha\in (1,2].
\end{equation}
As such, for any channel $\mathcal{N}_{A\to B}$, and any state $\rho_{RA^n}$, the following inequality holds:
\begin{align}
    E^{\varepsilon}_k\!\left(\mathcal{N}^{\otimes n}\!\left(\rho_{RA^n}\right)\right) &\le \widehat{E}^{\alpha}_k\!\left(\mathcal{N}^{\otimes n}\!\left(\rho_{RA^n}\right)\right) + \frac{\alpha}{\alpha - 1}\log_2\!\left(\frac{1}{1-\varepsilon}\right)\qquad \forall \alpha \in (1,2]\\
    &\le \widehat{E}^{\alpha}_k\!\left(\mathcal{N}^{\otimes n}\right) + \frac{\alpha}{\alpha - 1}\log_2\!\left(\frac{1}{1-\varepsilon}\right)\qquad \forall \alpha \in (1,2]\\
    &\le n\widehat{E}^{\alpha}_k\!\left(\mathcal{N}_{A\to B}\right) + \frac{\alpha}{\alpha - 1}\log_2\!\left(\frac{1}{1-\varepsilon}\right)\qquad \forall \alpha \in (1,2],\label{eq:hypo_ch_n_shot_le_geo_ch}
\end{align}
where the second inequality follows from Lemma~\ref{lem:unext_div_ch_ge_unext_div_st} and the final inequality follows from the subadditivity of the $k$-unextendible geometric R\'enyi divergence of channels. Substituting the inequality in~\eqref{eq:hypo_ch_n_shot_le_geo_ch} into~\eqref{eq:priv_cap_le_hypo_test_st} and using the monotonicity of the logarithm and exponential functions, we arrive at the following inequality:
\begin{align}
    P^{\varepsilon,\to}\!\left(\mathcal{N}_{A\to B}\right) &\le \log_2\!\left(\frac{k-1}{k}\right) - \log_2\!\left(2^{-n\widehat{E}^{\alpha}_k\left(\mathcal{N}\right) + \frac{\alpha}{\alpha - 1}\log_2\left(1-\varepsilon\right)} - \frac{1}{k}\right)\\
    &= \log_2\!\left(\frac{k-1}{k}\right) - \log_2\!\left(2^{-n\widehat{E}^{\alpha}_k\left(\mathcal{N}\right)}(1-\varepsilon)^{\frac{\alpha}{\alpha - 1}} - \frac{1}{k} \right).
\end{align}
This concludes the proof.

\section{Proof of Proposition~\ref{prop:weak_subadd_sandwich}}\label{app:weak_subadd_sandwich_proof}

In this section, we prove the weak subadditivity property of the $k$-unextendible sandwiched R\'enyi divergence of channels stated in Proposition~\ref{prop:weak_subadd_sandwich}.

We begin by showing that the $k$-unextendible sandwiched R\'enyi divergence of states is quasi-convex. 

\begin{proposition}\label{prop:sandwich_ent_quasi_convex}
    Let $\rho^1_{AB}$ and $\rho^2_{AB}$ be arbitrary bipartite states. Then,
    \begin{equation}
        \widetilde{E}^{\alpha}_k\!\left(\lambda \rho^1_{AB} + (1-\lambda)\rho^2_{AB}\right) \le \max_{x\in \{1,2\}}\widetilde{E}^{\alpha}_k\!\left(\rho^x_{AB}\right) \qquad \forall k\ge 2, \alpha >1, \lambda \in [0,1].
    \end{equation}
\end{proposition}

\begin{proof}
    Let us define the quasi-sandwiched R\'enyi relative entropy as follows:
\begin{equation}
\widetilde{Q}_{\alpha}(\rho\|\sigma)\coloneqq\Tr\!\left[\left(\sigma^{\frac{1-\alpha}{2\alpha}}\rho\sigma^{\frac{1-\alpha}{2\alpha}}\right)^{\alpha}\right].
\end{equation}
This function is known to be jointly convex in $\rho$ and $\sigma$
for $\alpha>1$ \cite[Proposition~3]{Frank2013}. We now prove that the function
\begin{equation}
\rho_{AB}\mapsto\widetilde{Q}_{\alpha,k}(\rho_{AB})\coloneqq\inf_{\sigma_{AB}\in\Ext_{k}(A:B)}\widetilde{Q}_{\alpha}(\rho_{AB}\|\sigma_{AB})
\end{equation}
is convex. To see this, pick arbitrary $\sigma_{AB}^{1},\sigma_{AB}^{2}\in\Ext_{k}(A:B)$
and consider that
\begin{align}
\widetilde{Q}_{\alpha,k}(\lambda\rho_{AB}^{1}+\left(1-\lambda\right)\rho_{AB}^{2}) & \leq\widetilde{Q}_{\alpha}(\lambda\rho_{AB}^{1}+\left(1-\lambda\right)\rho_{AB}^{2}\|\lambda\sigma_{AB}^{1}+\left(1-\lambda\right)\sigma_{AB}^{2})\\
 & \leq\lambda\widetilde{Q}_{\alpha}(\rho_{AB}^{1}\|\sigma_{AB}^{1})+\left(1-\lambda\right)\widetilde{Q}_{\alpha}(\rho_{AB}^{2}\|\sigma_{AB}^{2}),
\end{align}
where the first inequality follows because $\lambda\sigma_{AB}^{1}+\left(1-\lambda\right)\sigma_{AB}^{2}\in \Ext_{k}(A\!:\!B)$ and the second inequality follows from joint convexity.
Since the inequality holds for all $\sigma_{AB}^{1},\sigma_{AB}^{2}\in\Ext_{k}(A\!:\!B)$,
we conclude that
\begin{equation}
\widetilde{Q}_{\alpha,k}(\lambda\rho_{AB}^{1}+\left(1-\lambda\right)\rho_{AB}^{2})\leq\lambda\widetilde{Q}_{\alpha,k}(\rho_{AB}^{1})+\left(1-\lambda\right)\widetilde{Q}_{\alpha,k}(\rho_{AB}^{2}),
\end{equation}
thus justifying that $\rho_{AB}\mapsto\widetilde{Q}_{\alpha,k}(\rho_{AB}) $ is convex. Now consider that
\begin{align}
\widetilde{E}_{k}^{\alpha}(\lambda\rho_{AB}^{1}+\left(1-\lambda\right)\rho_{AB}^{2}) & =\frac{1}{\alpha-1}\log_{2}\widetilde{Q}_{\alpha,k}(\lambda\rho_{AB}^{1}+\left(1-\lambda\right)\rho_{AB}^{2})\\
 & \leq\frac{1}{\alpha-1}\log_{2}\left[\lambda\widetilde{Q}_{\alpha,k}(\rho_{AB}^{1})+\left(1-\lambda\right)\widetilde{Q}_{\alpha,k}(\rho_{AB}^{2})\right]\\
 & \leq\frac{1}{\alpha-1}\log_{2}\left[\max_{x\in\left\{ 1,2\right\} }\widetilde{Q}_{\alpha,k}(\rho_{AB}^{x})\right]\\
 & \leq\max_{x\in\left\{ 1,2\right\} }\frac{1}{\alpha-1}\log_{2}\widetilde{Q}_{\alpha,k}(\rho_{AB}^{x})\\
 & =\max_{x\in\left\{ 1,2\right\} }\widetilde{E}_{k}^{\alpha}(\rho_{AB}^{x}),
\end{align}
thus establishing the claim.
\end{proof}

\medskip

Now we continue with the proof of Proposition~\ref{prop:weak_subadd_sandwich}. In the remainder of this section, we fix $k\ge 2$, $n\in \mathbb{N}$, and $\alpha \in (1,\infty)$.

For every channel $\mathcal{N}_{A\to B}$, the tensor product channel $\mathcal{N}^{\otimes n}_{A\to B}$ is covariant under every permutation operation. That is,
\begin{equation}
    \mathcal{W}^{\pi}_{B^n}\circ\mathcal{N}^{\otimes n}_{A\to B} = \mathcal{N}^{\otimes n}_{A\to B}\circ\mathcal{W}^{\pi}_{A^n} \qquad \forall \pi \in S_k,
\end{equation}
where $\mathcal{W}^{\pi}$ is the permutation channel defined after~\eqref{eq:perm_cov_cond_k_ext}. Alternatively,
\begin{equation}
    \mathcal{W}^{\pi}_{B^n}\circ\mathcal{N}^{\otimes n}_{A\to B}\circ\mathcal{W}^{\pi^{-1}}_{A^n} = \mathcal{N}^{\otimes n}_{A\to B} \qquad \forall \pi \in S_k. \label{eq:tensor_ch_perm_cov}
\end{equation}

Let $\psi_{RA^n}$ be an arbitrary pure state. From~\eqref{eq:tensor_ch_perm_cov}, we can write
\begin{align}
    \widetilde{E}^{\alpha}_k\!\left(\mathcal{N}^{\otimes n}_{A\to B}\!\left(\psi_{RA^n}\right)\right) &= \widetilde{E}^{\alpha}_k\!\left(\mathcal{W}^{\pi}_{B^n}\circ\mathcal{N}^{\otimes n}_{A\to B}\circ\mathcal{W}^{\pi^{-1}}_{A^n}\!\left(\psi_{RA^n}\right)\right) \qquad \forall \pi \in S_k\\
    &= \widetilde{E}^{\alpha}_k\!\left(\mathcal{N}^{\otimes n}_{A\to B}\circ\mathcal{W}^{\pi^{-1}}_{A^n}\!\left(\psi_{RA^n}\right)\right) \qquad \forall \pi \in S_k,\label{eq:sandwich_n_ch_le_twirled_ch}
\end{align}
where the $k$-unextendible sandwiched R\'enyi divergence of states is computed with respect to the partition $R\!:\!B^n$. The equality in~\eqref{eq:sandwich_n_ch_le_twirled_ch} follows from the invariance of the $k$-unextendible sandwiched R\'enyi divergence of states under local unitary channels.

Let us define the following state:
\begin{equation}
    \overline{\rho}_{A^n} \coloneqq \frac{1}{|S_k|}\sum_{\pi \in S_k}\mathcal{W}^{\pi}_{A^n}\!\left(\operatorname{Tr}_R\!\left[\psi_{RA^n}\right]\right),
\end{equation}
and let $\psi^{\overline{\rho}}_{RA^n}$ be a purification of $\overline{\rho}_{A^n}$. Now applying reasoning similar to that in the proof of \cite[Proposition~2]{TWW17}, we conclude that
\begin{equation}\label{eq:sandwich_n_ch_le_avg_st}
    \widetilde{E}^{\alpha}_k\!\left(\mathcal{N}^{\otimes n}_{A\to B}\!\left(\psi_{RA^n}\right)\right) \le \widetilde{E}^{\alpha}_k\!\left(\mathcal{N}^{\otimes n}_{A\to B}\!\left(\psi^{\overline{\rho}}_{RA^n}\right)\right).
\end{equation}
Note that the state $\overline{\rho}_{A^n}$ is invariant under permutations, which implies that there exists a purification of $\overline{\rho}_{A^n}$, say $\phi^{\overline{\rho}}_{\hat{A}^nA^n}$, that is invariant under the channel $\mathcal{W}^{\pi}_{\hat{A}^n}\otimes \mathcal{W}^{\pi}_{A^n}$ (see~\cite[Lemma 4.2.2]{Ren06}). Since both $\psi^{\overline{\rho}}_{RA^n}$ and $\phi^{\overline{\rho}}_{\hat{A}^nA^n}$ are purifications of $\overline{\rho}_{A^n}$, there exists an isometric channel $\mathcal{P}_{\hat{A}^n\to R}$ such that
\begin{equation}\label{eq:pur_to_ext}
    \mathcal{P}_{\hat{A}^n\to R}\!\left(\phi^{\overline{\rho}}_{\hat{A}^nA^n}\right) = \psi^{\overline{\rho}}_{RA^n}.
\end{equation}
Combining~\eqref{eq:sandwich_n_ch_le_avg_st} and~\eqref{eq:pur_to_ext}, we arrive at the following inequality:
\begin{align}
    \widetilde{E}^{\alpha}_k\!\left(\mathcal{N}^{\otimes n}_{A\to B}\!\left(\psi_{RA^n}\right)\right) &\le \widetilde{E}^{\alpha}_k\!\left(\mathcal{N}^{\otimes n}_{A\to B}\otimes\mathcal{P}_{\hat{A}^n\to R}\!\left(\phi^{\overline{\rho}}_{\hat{A}^nA^n}\right)\right)\\
    &\le \widetilde{E}^{\alpha}_k\!\left(\mathcal{N}^{\otimes n}_{A\to B}\!\left(\phi^{\overline{\rho}}_{\hat{A}^nA^n}\right)\right),\label{eq:sandwich_le_avg_pure}
\end{align}
where the final inequality follows from the monotonicity of the $k$-unextendible sandwiched R\'enyi divergence of states under local channels.

Consider the projection onto the symmetric subspace of $\mathcal{H}^{\otimes n}_{\hat{A}A}$, which is defined as follows:
\begin{equation}
    \Pi^{\operatorname{sym}}_{\hat{A}^nA^n} = \frac{1}{|S_k|}\sum_{\pi \in S_k}W^{\pi}_{\hat{A}^n}\otimes W^{\pi}_{A^n},
\end{equation}
where $W^{\pi}$ is the unitary operator corresponding to the permutation $\pi$ in the symmetric group $S_k$. Since $\phi^{\overline{\rho}}_{\hat{A}^nA^n}$ is invariant under the channel $\mathcal{W}^{\pi}_{\hat{A}^n}\otimes \mathcal{W}^{\pi}_{A^n}$ for every $\pi\in S_k$,
\begin{equation}
    \Pi^{\operatorname{sym}}_{\hat{A}^nA^n}\phi^{\overline{\rho}}_{\hat{A}^nA^n}\Pi^{\operatorname{sym}}_{\hat{A}^nA^n} = \phi^{\overline{\rho}}_{\hat{A}^nA^n}.
\end{equation}
The fact that $\phi^{\overline{\rho}}_{\hat{A}^nA^n}$ is a state implies the following inequality:
\begin{align}
    \phi^{\overline{\rho}}_{\hat{A}^nA^n} &\le I_{\hat{A}^nA^n},\\
    \implies \Pi^{\operatorname{sym}}_{\hat{A}^nA^n}\phi^{\overline{\rho}}_{\hat{A}^nA^n}\Pi^{\operatorname{sym}}_{\hat{A}^nA^n} &\le \Pi^{\operatorname{sym}}_{\hat{A}^nA^n}I_{\hat{A}^nA^n}\Pi^{\operatorname{sym}}_{\hat{A}^nA^n},\\
    \implies \phi^{\overline{\rho}}_{\hat{A}^nA^n} & \le \Pi^{\operatorname{sym}}_{\hat{A}^nA^n},
\end{align}
where the second inequality follows from the positive semidefiniteness of $\Pi^{\operatorname{sym}}_{\hat{A}^nA^n}$ and the final inequality is a consequence of $\Pi^{\operatorname{sym}}_{\hat{A}^nA^n}$ being a projection operator.

As noted in~\cite[Proposition 6]{Har13}, the projection operator $\Pi^{\operatorname{sym}}_{\hat{A}^nA^n}$ can be written as follows:
\begin{equation}
    \Pi^{\operatorname{sym}}_{\hat{A}^nA^n} = \binom{n+|A|^2-1}{n}\int d\mu\!\left(\phi\right) \phi^{\otimes n}_{\hat{A}A},
\end{equation}
where $\mu(\phi)$ is the uniform probability distribution on the unit sphere consisting of pure bipartite states. As such,
\begin{equation}\label{eq:avg_pur_le_twirled_scaled}
    \phi^{\overline{\rho}}_{\hat{A}^nA^n} \le \binom{n+|A|^2-1}{n}\int d\mu\!\left(\phi\right) \phi^{\otimes n}_{\hat{A}A}.
\end{equation}

For all states $\omega^1_{AB}$ and $\omega^2_{AB}$ such that $\omega^1_{AB} \le \gamma\omega^2_{AB}$ for some $\gamma \ge 1$, the following inequality holds:
\begin{equation}\label{eq:sandwich_ent_gamma_scaled}
    \widetilde{E}^{\alpha}_k\!\left(\omega^1_{AB}\right) \le \frac{\alpha}{\alpha - 1}\log_2\gamma + \widetilde{E}^{\alpha}_k\!\left(\omega^2_{AB}\right),
\end{equation}
which follows directly from~\cite[Lemma 5]{TWW17}. Combining~\eqref{eq:sandwich_le_avg_pure},~\eqref{eq:avg_pur_le_twirled_scaled}, and~\eqref{eq:sandwich_ent_gamma_scaled}, we arrive at the following inequality:
\begin{align}
    \widetilde{E}^{\alpha}_k\!\left(\mathcal{N}^{\otimes n}_{A\to B}\!\left(\psi_{RA^n}\right)\right) &\le \frac{\alpha}{\alpha - 1}\log_2\!\left(\binom{n+|A|^2-1}{n}\right) + \widetilde{E}^{\alpha}_k\!\left(\int d\mu(\phi)\mathcal{N}^{\otimes n}_{A\to B}\!\left(\phi^{\otimes n}_{\hat{A}A}\right)\right)\\
    &\le \frac{\alpha}{\alpha - 1}\log_2\!\left(\binom{n+|A|^2-1}{n}\right) + \sup_{\phi_{\hat{A}A}}\widetilde{E}^{\alpha}_k\!\left(\mathcal{N}^{\otimes n}_{A\to B}\!\left(\phi^{\otimes n}_{\hat{A}A}\right)\right),\label{eq:sandwich_n_ch_le_sandwich_tensor_st}
\end{align}
where the final inequality follows from Proposition~\ref{prop:sandwich_ent_quasi_convex}, and the supremum in the final inequality is over all bipartite pure states in $\mathcal{S}(\hat{A}A)$.

Recall that the $k$-unextendible sandwiched R\'enyi divergence of a state is subadditive with respect to tensor products. Therefore,
\begin{align}
    \widetilde{E}^{\alpha}_k\!\left(\mathcal{N}^{\otimes n}_{A\to B}\!\left(\phi^{\otimes n}_{\hat{A}A}\right)\right) &= \widetilde{E}^{\alpha}_k\!\left(\left(\mathcal{N}_{A\to B}\!\left(\phi_{\hat{A}A}\right)\right)^{\otimes n}\right)\\
    &\le n\widetilde{E}^{\alpha}_k\!\left(\mathcal{N}_{A\to B}\!\left(\phi_{\hat{A}A}\right)\right).
\end{align}
Substituting the above inequality in~\eqref{eq:sandwich_n_ch_le_sandwich_tensor_st}, we arrive at the following:
\begin{align}
    \widetilde{E}^{\alpha}_k\!\left(\mathcal{N}^{\otimes n}_{A\to B}\!\left(\psi_{RA^n}\right)\right) &\le \frac{\alpha}{\alpha - 1}\log_2\!\left(\binom{n+|A|^2-1}{n}\right) + \sup_{\phi_{\hat{A}A}}n\widetilde{E}^{\alpha}_k\!\left(\mathcal{N}_{A\to B}\!\left(\phi_{\hat{A}A}\right)\right)\\
    &\le \frac{\alpha}{\alpha - 1}\log_2\!\left(\binom{n+|A|^2-1}{n}\right) + n\widetilde{E}^{\alpha}_k\!\left(\mathcal{N}_{A\to B}\right),\label{eq:sandwich_weak_subadd_final}
\end{align}
where the final inequality follows from Lemma~\ref{lem:unext_div_ch_ge_unext_div_st}. 

Now consider an arbitrary state $\omega_{RA^n}$ with the following pure-state decomposition:
\begin{equation}
    \omega_{RA^n} = \sum_i\lambda_i\varphi^i_{RA^n}.
\end{equation}
Then, from the quasi-convexity of $k$-unextendible sandwiched R\'enyi divergence of states, 
\begin{align}
    \widetilde{E}^{\alpha}_k\!\left(\mathcal{N}^{\otimes n}_{A\to B}\!\left(\omega_{RA^n}\right)\right) &\le \max_{i}\widetilde{E}^{\alpha}_k\!\left(\mathcal{N}^{\otimes n}_{A\to B}\!\left(\varphi^i_{RA^n}\right)\right)\\
    &\le \frac{\alpha}{\alpha - 1}\log_2\!\left(\binom{n+|A|^2-1}{n}\right) + n\widetilde{E}^{\alpha}_k\!\left(\mathcal{N}_{A\to B}\right),\label{eq:sandwich_weak_subadd_final_mixed}
\end{align}
where the final inequality follows from~\eqref{eq:sandwich_weak_subadd_final}. 

Since~\eqref{eq:sandwich_weak_subadd_final_mixed} holds for every state $\omega_{RA^n}$, every $k\ge 2$, and every $\alpha>1$, we conclude the statement of Proposition~\ref{prop:weak_subadd_sandwich}.

\section{Proof of Corollary~\ref{cor:n_shot_priv_cap_sandwich}}\label{app:n_shot_priv_cap_sandwich_proof}

In this section, we find a single-letter upper bound on the $n$-shot forward-assisted private capacity of a channel in terms of the $k$-unextendible sandwiched R\'enyi divergence of the channel.

Fix $k\ge 2$ and $\alpha > 1$. Let $\mathcal{N}_{A\to B}$ be an arbitrary channel. We first substitute the inequality from~\eqref{eq:hypo_test_ent_le_sandwich} into~\eqref{eq:priv_cap_le_hypo_test_st}, which leads to the following inequality:
\begin{align}
    P^{\varepsilon,\to}\!\left(\mathcal{N}_{A\to B}\right) &\le \log_2\!\left(\frac{k-1}{k}\right) - \log_2\!\left(2^{-\sup_{\rho_{RA}\in \mathcal{S}(RA)}\left\{\widetilde{E}^{\alpha}_k\left(\mathcal{N}(\rho_{RA})\right) + \frac{\alpha}{\alpha - 1}\log_2\left(\frac{1}{1-\varepsilon}\right)\right\}} - \frac{1}{k}\right)\\
    &= \log_2\!\left(\frac{k-1}{k}\right) - \log_2\!\left(2^{-\sup_{\rho_{RA}\in \mathcal{S}(RA)}\widetilde{E}^{\alpha}_k\left(\mathcal{N}(\rho_{RA})\right)}\!\left(1-\varepsilon\right)^{\frac{\alpha}{\alpha - 1}} - \frac{1}{k}\right).
\end{align}
As such, for a tensor product of $n$ copies of $\mathcal{N}_{A\to B}$, we have the following inequality:
\begin{equation}\label{eq:priv_cap_sandwich_tensor_prod}
    P^{\varepsilon,\to}\!\left(\mathcal{N}^{\otimes n}_{A\to B}\right) \le \log_2\!\left(\frac{k-1}{k}\right) - \log_2\!\left(2^{-\sup_{\rho_{RA^n}\in \mathcal{S}(RA^n)}\widetilde{E}^{\alpha}_k\left(\mathcal{N}^{\otimes n}(\rho_{RA^n})\right)}\!\left(1-\varepsilon\right)^{\frac{\alpha}{\alpha - 1}} - \frac{1}{k}\right).
\end{equation}

Since the statement of Proposition~\ref{prop:weak_subadd_sandwich} holds for every state $\rho_{RA^n}$, we can take a supremum over all states in $\mathcal{S}(RA^n)$ and arrive at the following inequality:
\begin{equation}\label{eq:sup_ov_st_sandwich_tensor_ch}
    \sup_{\rho_{RA^n}\in \mathcal{S}(RA^n)}\widetilde{E}^{\alpha}_k\!\left(\mathcal{N}^{\otimes n}(\rho_{RA^n})\right) \le \frac{\alpha}{\alpha - 1}\log_2\!\left(C(n,|A|)\right) + n\widetilde{E}^{\alpha}_k\!\left(\mathcal{N}_{A\to B}\right),
\end{equation}
where
\begin{equation}
    C(n,|A|) \coloneqq \binom{n+|A|^2-1}{n}.
\end{equation}
Substituting~\eqref{eq:sup_ov_st_sandwich_tensor_ch} into~\eqref{eq:priv_cap_sandwich_tensor_prod}, we arrive at the following inequality:
\begin{align}
    P^{\varepsilon,\to}\!\left(\mathcal{N}^{\otimes n}_{A\to B}\right) &\le \log_2\!\left(\frac{k-1}{k}\right) - \log_2\!\left(2^{-\left\{n\widetilde{E}^{\alpha}_k\left(\mathcal{N}\right) + \frac{\alpha}{\alpha - 1}\log_2\left(C(n,|A|)\right)\right\}}\!\left(1-\varepsilon\right)^{\frac{\alpha}{\alpha - 1}} - \frac{1}{k}\right)\\
    &= \log_2\!\left(\frac{k-1}{k}\right) - \log_2\!\left(2^{-n\widetilde{E}^{\alpha}_k\left(\mathcal{N}\right)}\!\left(\frac{1-\varepsilon}{C(n,|A|)}\right)^{\frac{\alpha}{\alpha - 1}} - \frac{1}{k}\right).
\end{align}
This concludes the proof.

\section{Proof of Proposition~\ref{prop:eras_ch_gen_div}}\label{app:eras_ch_gen_div}

In this section, we find an upper bound on the $k$-unextendible generalized channel divergence of a tensor product of $n$ erasure channels in terms of the generalized divergence between two probability distributions.

Let $\left\{U^g_A\right\}_{g\in G}$ be a unitary one-design acting on the Hilbert space $\mathcal{H}_A$. Let us define the following operators:
\begin{equation}
    V^g \coloneqq U^g + |e\rangle\!\langle e| \qquad \forall g\in \mathcal{G}.
\end{equation}
Note that all erasure channels are covariant with respect to $\left\{\left(U^g,V^g\right)\right\}_{g\in \mathcal{G}}$. That is,
\begin{equation}
   \mathcal{E}^p\!\left(U^g\rho \left(U^g\right)^{\dagger}\right) = V^g\mathcal{E}^p\!\left(\rho \right)\!\left(V^g\right)^{\dagger} \qquad \forall g\in \mathcal{G}, p\in [0,1],
\end{equation}
where $\mathcal{E}^p$ is an erasure channel with an erasure probability $p$, defined in~\eqref{eq:eras_ch_defn}. 

Let $h\coloneqq (g_1(h),g_2(h),\ldots,g_n(h))$ be an arbitrary element of the set $\mathcal{G}^{\times n}$, where $g_i\in \mathcal{G}$ for every $i\in \{1,2,\ldots,n\}$. Note that $\left\{U^{g_1(h)}\otimes U^{g_2(h)}\otimes \cdots\otimes U^{g_n(h)}\right\}_{h\in \mathcal{G}^{\times n}}$ is also a unitary one-design on the Hilbert space $\mathcal{H}^{\otimes n}_{A}$ for any positive integer $n$. Let us use the following notations:
\begin{align}
    U^h &\coloneqq U^{g_1(h)}\otimes U^{g_2(h)}\otimes \cdots\otimes U^{g_n(h)},\\
    V^h &\coloneqq V^{g_1(h)}\otimes V^{g_2(h)}\otimes \cdots\otimes V^{g_n(h)}.
\end{align}
Then it can be verified that a tensor product of erasure channels $\left(\mathcal{E}^p\right)^{\otimes n}$ is covariant with respect to $\left\{\left(U^h,V^h\right)\right\}_{h\in \mathcal{G}^{\times n}}$. Since $\left\{U^h\right\}_{h\in \mathcal{G}^{\times n}}$ is a unitary one-design, we can use~\cite[Corollary II.5]{LKDW18} to arrive at the following equality:
\begin{equation}\label{eq:eras_ch_eq_eras_st}
    \mathbf{D}\!\left(\left(\mathcal{E}^p_{A\to B}\right)^{\otimes n}\middle\Vert\left(\mathcal{E}^q_{A\to B}\right)^{\otimes n}\right) = \mathbf{D}\!\left(\left(\mathcal{E}^p_{A\to B}\right)^{\otimes n}\!\left(\Phi_{R^nA^n}\right)\middle\Vert\left(\mathcal{E}^q_{A\to B}\right)^{\otimes n}\!\left(\Phi_{R^nA^n}\right)\right) \qquad \forall p,q \in [0,1],
\end{equation}
where $\Phi_{R^nA^n}$ is the maximally entangled state on the Hilbert space $\mathcal{H}^{\otimes n}_R\otimes \mathcal{H}^{\otimes n}_A$.

Let us define the following projections:
\begin{align}
    P^0_B &= \sum_{i=0}^{|B|-1}|i\rangle\!\langle i|_B\\
    P^1_B &= |e\rangle\!\langle e|_B.
\end{align}
Now consider the following POVM:
\begin{equation}
    \Pi_{B_{[n]}}\coloneqq \left\{\bigotimes_{j=1}^nP^{i_j}_{B_j}\right\}_{(i_1,i_2,\ldots,i_n)\in \{0,1\}^n}.
\end{equation}
The POVM $\Pi_{B_{[n]}}$ counts the number of erasure symbols in a state on $\mathcal{H}^{\otimes n}_B$ without modifying the state. When acted upon $\left(\mathcal{E}^p_{A\to B}\right)^{\otimes n}\!\left(\Phi_{R^nA^n}\right)$, the outcome of the POVM $\Pi_{B_{[n]}}$ is distributed with respect to the binomial distribution $\left\{1-p,p\right\}^{\times n}$. The data-processing inequality of generalized divergence thus implies the following:
\begin{equation}
    \mathbf{D}\!\left(\left\{1-p,p\right\}^{\times n}\middle \Vert \left\{1-q,q\right\}^{\times n}\right) \le \mathbf{D}\!\left(\left(\mathcal{E}^p_{A\to B}\right)^{\otimes n}\!\left(\Phi_{R^nA^n}\right)\middle\Vert\left(\mathcal{E}^q_{A\to B}\right)^{\otimes n}\!\left(\Phi_{R^nA^n}\right)\right).
\end{equation}
Also, note that one can construct the state $\left(\mathcal{E}^p_{A\to B}\right)^{\otimes n}\!\left(\Phi_{R^nA^n}\right)$ if one has access to the binomial distribution $\left\{1-p,p\right\}^{\times n}$. One way to achieve this is by first generating the state $\Phi_{R^nB^n}$, then drawing $n$ bits from the distribution $\left\{1-p,p\right\}^{\times n}$ and erasing the state on $B_i$ if the $i^{\operatorname{th}}$ draw corresponds to the outcome that occurs with probability $p$. The data-processing inequality for the generalized divergence now yields the opposite inequality:
\begin{equation}
    \mathbf{D}\!\left(\left\{1-p,p\right\}^{\times n}\middle \Vert \left\{1-q,q\right\}^{\times n}\right) \ge \mathbf{D}\!\left(\left(\mathcal{E}^p_{A\to B}\right)^{\otimes n}\!\left(\Phi_{R^nA^n}\right)\middle\Vert\left(\mathcal{E}^q_{A\to B}\right)^{\otimes n}\!\left(\Phi_{R^nA^n}\right)\right).
\end{equation}
Therefore,
\begin{equation}\label{eq:bin_eq_eras_st}
    \mathbf{D}\!\left(\left\{1-p,p\right\}^{\times n}\middle \Vert \left\{1-q,q\right\}^{\times n}\right) = \mathbf{D}\!\left(\left(\mathcal{E}^p_{A\to B}\right)^{\otimes n}\!\left(\Phi_{R^nA^n}\right)\middle\Vert\left(\mathcal{E}^q_{A\to B}\right)^{\otimes n}\!\left(\Phi_{R^nA^n}\right)\right).
\end{equation}

Now we turn our attention to the $k$-unextendible generalized divergence of a tensor product of erasure channels. Recall that $\mathcal{E}^{1-1/k}_{A\to B}$ is a $k$-extendible channel, and consequently, $\left(\mathcal{E}^{1-1/k}_{A\to B}\right)^{\otimes n}$ is also a $k$-extendible channel. Then, by definition of the $k$-unextendible generalized divergence of channels,
\begin{align}
    \mathbf{E}_k\!\left(\left(\mathcal{E}^{p}_{A\to B}\right)^{\otimes n}\right) &\le \mathbf{D}\!\left(\left(\mathcal{E}^p_{A\to B}\right)^{\otimes n}\middle\Vert\left(\mathcal{E}^{1-1/k}_{A\to B}\right)^{\otimes n}\right)\\
    &= \mathbf{D}\!\left(\left\{1-p,p\right\}^{\times n}\middle \Vert \left\{\frac{1}{k},1-\frac{1}{k}\right\}^{\times n}\right),
\end{align}
where the final equality follows from~\eqref{eq:eras_ch_eq_eras_st} and~\eqref{eq:bin_eq_eras_st}.

\bibliographystyle{alphaurl}
\bibliography{Ref}

\end{document}